\documentclass[a4paper,UKenglish]{lipics}

\usepackage{cite}
\usepackage{microtype}
\usepackage{latexsym}                            % LaTeX 2e
\usepackage{t1enc}
\usepackage{xspace}                               % LaTeX 2e
\usepackage{pdfpages}
\usepackage{color}
\usepackage{ifpdf}
\usepackage{graphicx}
\usepackage{pgf}
\usepackage{tikz}
\usepackage{pstricks}
\usepackage{mathrsfs}
\usepackage{amsfonts}% ``fleqn'' means formula are aligned to the left; ``leqno'' means the index of formula will appear to the left.
\usepackage{amsmath}
\usepackage{amssymb}
\usepackage{stmaryrd}
\usepackage{verbatim}
\usepackage{dsfont}
\usepackage{wasysym}
\usepackage{esvect}
\usepackage{url}
\usepackage{subfig}
\usepackage{paralist}
\usepackage{tabularx}
\usepackage{listings}
\usepackage{epic}

\title{Confluence of layered rewrite systems}

\author[1,2]{Jiaxiang Liu}
\author[2]{Jean-Pierre Jouannaud}
\author[3]{Mizuhito Ogawa}

\affil[1]{School of Software, and TNList, Tsinghua University, Beijing, China}
\affil[2]{Deducteam, INRIA, and LIX, \'{E}cole Polytechnique, Palaiseau, France}
\affil[3]{School of Information Science, JAIST, Ishikawa, Japan}
\authorrunning{J. Liu, J.-P. Jouannaud and M. Ogawa} %mandatory. First: Use abbreviated first/middle names. Second (only in severe cases): Use first author plus 'et. al.'
\Copyright{Jiaxiang Liu, Jean-Pierre Jouannaud and Mizuhito Ogawa}%mandatory, please use full first names. LIPIcs license is "CC-BY";  http://creativecommons.org/licenses/by/3.0/


\serieslogo{}%please provide filename (without suffix)
\volumeinfo%(easychair interface)
  {Billy Editor and Bill Editors}% editors
  {2}% number of editors: 1, 2, ....
  {Conference title on which this volume is based on}% event
  {1}% volume
  {1}% issue
  {1}% starting page number
\EventShortName{}
\DOI{10.4230/LIPIcs.xxx.yyy.p}% to be completed by the volume editor
\begin{document}

\maketitle
\thispagestyle{empty}
%%%%%%%%%%%%%%%%%%%%%%%%%%%%%%%%%%%%%%%%%%%%%%%%%%%%%%%%%%%%%%%%%%%%
%%%%%%%%%%%%%%%%%%%%%%%%%%%%%%%%%%%%%%%%%%%%%%%%%%%%%%%%%%%%%%%%%%%%
\newcommand{\Note}[1]{#1}
        % typiquement pre'vu pour pouvoir creer un index des notations

%%%%%%%%%%%%%%%%%%%%%%%%%%%%%%%%%%%%%%%%%%%%%%%%%%%%%%%%%%%%%%%%%%%%
%\newcommand{\intersection}{\cap}
\newcommand{\Not}{\neg}
\newcommand{\Or}{\vee}
\newcommand{\AND}[2]{\mathop{\bigwedge_{#1}^{#2}}}
\newcommand{\OR}[2]{\mathop{\bigvee_{#1}^{#2}}}
\newcommand{\bottom}{\bot}
\newcommand{\vide}{\emptyset}
\newcommand{\cqfd}{\hfill $\Box$}
\newcommand{\lin}[1]{\overline{#1}}
%\newcommand{\qed}{\hfill \Box}
%%%%%%%%%%%%%%%%%%%%%%%%%%%%%%%%%%%%%%%%%%%%%%%%%%%%%%%%%%%%%%%%%%%%

%% Acryonyms
\newcommand{\Coq}{{\sc Coq}\xspace}
\newcommand{\CoqMT}{{\sc CoqMT}\xspace}
\newcommand{\CoqMTU}{{\sc CoqMTU}\xspace}

%%%%%%%%%%%%%%%%%%%%%%%%%%%%%%%%%%%%%%%%%%%%%%%%%%%%%%%%%%%%%%%%%%%%
%% Fonts

\newcommand{\cA}{\mathcal{A}}
\newcommand{\cB}{\mathcal{B}}
\newcommand{\cC}{\mathcal{C}}
\newcommand{\cD}{\mathcal{D}}
\newcommand{\cE}{\mathcal{E}}
\newcommand{\cF}{\mathcal{F}}
\newcommand{\cI}{\mathcal{I}}
\newcommand{\cJ}{\mathcal{J}}
\newcommand{\cK}{\mathcal{K}}
\newcommand{\cL}{\mathcal{L}}
\newcommand{\cM}{\mathcal{M}}
\newcommand{\cN}{\mathcal{N}}
\newcommand{\cO}{\mathcal{O}}
\newcommand{\cP}{\mathcal{P}}
\newcommand{\cQ}{\mathcal{Q}}
\newcommand{\cR}{\mathcal{R}}
\newcommand{\cS}{\mathcal{S}}
\newcommand{\cT}{\mathcal{T}}
\newcommand{\cU}{\mathcal{U}}
\newcommand{\cV}{\mathcal{V}}
\newcommand{\cW}{\mathcal{W}}
\newcommand{\cX}{\mathcal{X}}
\newcommand{\cY}{\mathcal{Y}}

\newcommand{\der}{\longleftarrow\!\!}
\newcommand{\lrder}[1]{\mathop{\longrightarrow\!\!\!\!\!\ra}^{#1}}
\newcommand{\lrdertwo}[2]{\mathop{\longrightarrow\!\!\!\!\!\ra\!}^{#1}_{#2}}
\newcommand{\dlrder}[2]{\displaystyle \lrdertwo{#1}{#2}}
\newcommand{\rlder}[1]{~^{#1}\!\!\mathop{\la\!\!\!\!\!\longleftarrow}}
\newcommand{\rldertwo}[2]{\!~^{#1}_{#2}\!\!\!\mathop{\la\!\!\!\!\!\longleftarrow}}
\newcommand{\drlder}[2]{\displaystyle \rldertwo{#1}{#2}}
%\newcommand{\drlder}[1]{\mathop{\!\!\la\longleftarrow^{#1}}}

%labelled rewriting
%Optional Args 1: position, 2:label, 3:rule or rule set, 4:substitution.
\newcommand{\lablrps}[4]{\mathop{\longrightarrow}^{#1,#2}_{#3,#4}}
\newcommand{\dlablrps}[4]{\displaystyle \lablrps{#1}{#2}{#3}{#4}}
\newcommand{\lablrpps}[4]{\mathop{\Longrightarrow}^{#1,#2}_{#3,#4}}
\newcommand{\dlablrpps}[4]{\displaystyle \lablrpps{#1}{#2}{#3}{#4}}

% rewriting
\newcommand{\lrps}[2]{\mathop{\longrightarrow}^{#1}_{#2}}
\newcommand{\dlrps}[2]{\displaystyle \lrps{#1}{#2}}
\newcommand{\rlps}[2]{~{}^{#1}_{#2}\!\!\mathop{\longleftarrow}}
\newcommand{\rlpsbis}[2]{\mathop{\longleftarrow}^{#1}_{#2}}
\newcommand{\drlps}[2]{\displaystyle \rlpsbis{#1}{#2}}

% parallel rewriting
\newcommand{\Lrps}[2]{\mathop{\Longrightarrow}^{#1}_{#2}}
\newcommand{\dLrps}[2]{\displaystyle {\Lrps{#1}{#2}}}
\newcommand{\Rlps}[2]{\mathop{\Longleftarrow}^{#1}_{#2}}
\newcommand{\dRlps}[2]{\displaystyle {\Rlps{#1}{#2}}}
\newcommand{\Lrpseq}[2]{\mathop{\Longrightarrow\!\!\!\!\!\!\!-}^{#1}_{#2}}
\newcommand{\dLrpseq}[2]{\displaystyle {\Lrpseq{#1}{#2}}}
\newcommand{\Rlpseq}[2]{\mathop{-\!\!\!\!\!\!\!\Longleftarrow}^{#1}_{#2}}
\newcommand{\dRlpseq}[2]{\displaystyle {\Rlpseq{#1}{#2}}}

%% symmetric closures
\newcommand{\eqps}[1]{\mathop{\leftrightarrow}^{#1}}
\newcommand{\deqps}[1]{\displaystyle \eqps{#1}}
\newcommand{\leqps}[2]{~\eqps{#1}{#2}~}
\newcommand{\leqpps}[2]{~\Eqps{#1}{#2}~}

\newcommand{\convv}[2]
           {\mathop{\la\!\!\!\!\!\leftrightarrow\!\!\!\!\!\ra}^{#1}_{#2}}
\newcommand{\dconv}[1]{\displaystyle \convv{#1}{}}
\newcommand{\dconvtwo}[2]{\displaystyle \convv{#1}{#2}}

\newcommand{\oc}{>_{oc}} 

%% reflexive closures
\newcommand{\lrpseq}[1]{\mathop{\longrightarrow\!\!\!\!\!\!\!{\tiny =}}^{#1}}
\newcommand{\rlpseq}[1]{\mathop{{}^{#1}{\tiny =}\!\!\!\!\!\!\!\longleftarrow}}
\newcommand{\lrpseqtwo}[2]{\mathop{\longrightarrow\!\!\!\!\!\!\!{\tiny =}}^{#1}_{#2}}
\newcommand{\rlpseqtwo}[2]{\mathop{{\tiny =}\!\!\!\!\!\!\!\longleftarrow}^{#1}_{#2}}
\newcommand{\lrppseq}[1]{\mathop{\Longrightarrow\!\!\!\!\!\!\!{\tiny =}}^{#1}}
\newcommand{\rlppseq}[1]{\mathop{{\tiny =}\!\!\!\!\!\!\!\Longleftarrow}^{#1}}
\newcommand{\dlrpseq}[2]{{\displaystyle \lrpseqtwo{#1}{#2}}}
\newcommand{\drlpseq}[2]{{\displaystyle \rlpseqtwo{#1}{#2}}}

%\newcommand{\deqps}[1]{\displaystyle \eqps{#1}}

%% others
\newcommand{\Eqps}[2]{\mathop{\Longleftrightarrow}^{#1}_{#2}}
\newcommand{\LrpsD}{\Longrightarrow_\cD}
\newcommand{\RlpsD}{_\cD\!\!\Longleftarrow}

\newcommand{\la}{\leftarrow}
\newcommand{\ra}{\rightarrow}
\newcommand{\da}{\downarrow}
\newcommand{\ua}{\uparrow}
\newcommand{\lra}{\leftrightarrow}
\newcommand{\lh}{\leftharpoons}
\newcommand{\rh}{\rightharpoons}
\newcommand{\rlh}{\rightleftharpoons}
\newcommand{\Lla}{\longleftarrow}
\newcommand{\Lra}{\longrightarrow}
\newcommand{\Llra}{\longleftrightarrow}
\newcommand{\Dla}{\Leftarrow}
\newcommand{\Dra}{\Rightarrow}
\newcommand{\Dda}{\Downarrow}
\newcommand{\Dua}{\Uparrow}
\newcommand{\Dlra}{\Leftrightarrow}
\newcommand{\Nla}{\not\rightarrow}
\newcommand{\Nra}{\not\leftarrow}
\newcommand{\raa}{\ra_{aliens}}
\newcommand{\rac}{\ra_{cap}}
\newcommand{\raA}{\ra_{A}}
\newcommand{\raC}{\ra_{C}}

\newcommand{\TFX}{\cT(\cF,\cX)}
\newcommand{\GTF}{\cT(\cF)}

\newcommand{\Head}[1]{{\cH}ead({#1})}
\newcommand{\Var}[1]{{\cV}ar({#1})}
\newcommand{\FVar}[1]{{\cFV}ar({#1})}
\newcommand{\Pos}[1]{{\cP}os({#1})}
\newcommand{\FPos}[1]{{\cF\cP}os({#1})}
\newcommand{\VPos}[1]{{\cV\cP}os({#1})}
\newcommand{\Posv}[2]{{\cP}os_{#1}(#2)}
\newcommand{\rootp}{\mbox {\footnotesize $\Lambda$}}
%root position
\newcommand{\para}{\parallel}
\newcommand{\Dom}[1]{{\cD}om({#1})}
\newcommand{\FDom}[1]{{\cFD}om({#1})}
\newcommand{\Ran}[1]{{\cR}an({#1})}

%Ordres
%greater than partial order
\newcommand{\gtpo}{\succ}
%greater equal partial order
\newcommand{\gepo}{\succeq}
%less than partial order
\newcommand{\ltpo}{\prec}
%less equal partial order
\newcommand{\lepo}{\preceq}
%equivalence associated with the quasi partial order
\newcommand{\eqpo}{\simeq}

\newcommand{\gtsubt}{\rhd}
\newcommand{\gesubt}{\unrhd}
\newcommand{\gsubt}{\unrhd}
%subterm (greater version)
\newcommand{\sgsubt}{\rhd}
%strict subterm (greater)
\newcommand{\lsubt}{\unlhd}
\newcommand{\ltsubt}{\lhd}
\newcommand{\lesubt}{\unlhd}
%subterm (less version)
\newcommand{\slsubt}{\lhd}
%strict subterm (less)
\newcommand{\gesubs}{\mathop{\stackrel{\scriptscriptstyle\bullet}{}\!\!\!\geq}}
\newcommand{\gsubsq}{\q{\gsubs}}

\newcommand{\nf}[4]{#1{{}_{#3}^{}\!\!\downarrow^{#2}_{#4}}}

%%%%%%%%%%%%%%%%%%%%%%%%%%%%%%%%%%

\newcommand{\hide}[1]{\ignorespaces}
\renewcommand{\displaystyle}{}
\renewcommand{\subsection}[1]{\paragraph*{\bf #1}}
\renewcommand{\footnote}[1]{ (#1)}
\renewcommand{\longrightarrow}{\rightarrow}
\renewcommand{\Longrightarrow}{\Rightarrow}
\renewcommand{\longleftarrow}{\leftarrow}
\renewcommand{\Longleftarrow}{\Leftarrow}
\let\ep=\endproof
%\renewcommand{\endproof}{\qed\ep}

%%% macros specifiques
\newcommand{\jp}[1]{{\bf JP:}\quad #1 \quad{\bf end}}
\newcommand{\jpj}[1]{{\bf JP:}\quad #1 \quad{\bf end}}
\newcommand{\jx}[1]{{\bf Jiaxiang:}\quad #1 \quad{\bf end}}
\newcommand{\mo}[1]{{\bf Mizuhito:}\quad #1 \quad{\bf end}}

\newcommand{\cFD}{\cF\cD}
\newcommand{\bh}{\cB\cH}
\newcommand{\family}{\cF y}
\newcommand{\vect}[1]{\vec{#1}}
\newcommand{\nat}{\mbox{$I\!\!N$}}
\newcommand{\join}{\rhd\!\!\!\lhd}
\newcommand{\eq}{\stackrel{?}{=}}
\newcommand{\ov}[1]{\overline{#1}}

%vocabularies
\newcommand{\Rt}{\ensuremath{R_{\textit{T}}}\xspace}
\newcommand{\Rnt}{\ensuremath{R_{\textit{NT}}}\xspace}
\newcommand{\Rc}{\mbox{$R_{C}$}\xspace}
\newcommand{\nfRc}[1]{\, #1 \!\!\da_{\Rc}}
\newcommand{\nfRt}[1]{\, #1 \!\!\da_{\Rt}}
\newcommand{\R}{\ensuremath{R}\xspace}
\newcommand{\Rtmod}{\ensuremath{{\Rt}_{sub}}\xspace}
\newcommand{\Rmod}{\ensuremath{\Rtmod\cup\Rnt}\xspace}
\newcommand{\Rsub}{\ensuremath{R_{\textit{R}}}}

\newcommand{\Ft}{\mbox{$F_{\textit{T}}$}\xspace}
\newcommand{\Fnt}{\mbox{$F_{\textit{NT}}$}\xspace}
\newcommand{\F}{\mbox{$F$}\xspace}
\newcommand{\Fc}{\mbox{$F_{\textit{C}}$}\xspace}
\newcommand{\Ftc}{\mbox{$F_{\textit{T} \setminus \textit{C}}$}\xspace}
\newcommand{\Fntc}{\mbox{$F_{\textit{NT} \setminus \textit{C}}$}\xspace}
\newcommand{\Tt}{\mbox{$\cT(\Ft,\cX)$}\xspace}
\newcommand{\TtY}{\ensuremath{\cT(\Ft,\cX\cup\cY)}\xspace}
\newcommand{\Tnt}{\mbox{$\cT(\Fnt,\cX)$}\xspace}
\newcommand{\TntY}{\ensuremath{\cT(\Fnt,\cX\cup\cY)}\xspace}
\newcommand{\Tc}{\mbox{$\cT(\Fc,\cX)$}\xspace}
\newcommand{\TcY}{\ensuremath{\cT(\Fc,\cX\cup\cY)}\xspace}
\newcommand{\T}{\mbox{$\cT(\F,\cX)$}\xspace}
\newcommand{\Ty}{\ensuremath{\cT(\F,\cX\cup\cY)}\xspace}
\newcommand{\Tn}{\mbox{$\cT_n(\cF,\cX)$}\xspace}

\newcommand{\po}{\mbox{$\succ\!\!\succ$}}
\newcommand{\ordo}{\succ}

%labels
\newcommand{\gtordl}{\rhd}
\newcommand{\ordl}{\gtordl}
\newcommand{\geordl}{\unrhd}
\newcommand{\ltordl}{\lhd}
\newcommand{\leordl}{\unlhd}

%positions
\newcommand{\gtordp}{>}
\newcommand{\ordp}{\gtordp}
\newcommand{\geordp}{\geq}
\newcommand{\ltordp}{<}
\newcommand{\leordp}{\leq}

%others
\newcommand{\gtsubs}{\gtrdot}   %{{\scriptsize \bullet}{\hspace{-0.75mm}>}}
\newcommand{\gtnat}{>}
\newcommand{\ltnat}{<}
\newcommand{\genat}{\geq}
\newcommand{\lenat}{\leq}

%aliens, caps and hats
\renewcommand{\t}{{\sc t}\xspace}
\newcommand{\nt}{{\sc nt}\xspace}
\renewcommand{\it}{{\sc T}\xspace}
\renewcommand{\int}{{\sc NT}\xspace}

\newcommand{\CPost}[1]{{\cal CP}os_\it({#1})}
\newcommand{\CPosnt}[1]{{\cal CP}os_\int({#1})}
\newcommand{\APost}[1]{{\cal AP}os_\it({#1})}
\newcommand{\APosnt}[1]{{\cal AP}os_\int({#1})}
\newcommand{\CPos}[1]{{\cal CP}os({#1})}
\newcommand{\APos}[1]{{\cal AP}os({#1})}
\newcommand{\rank}[1]{rk({#1})}

\newcommand{\myhat}[1]{\widehat{#1}}
\newcommand{\myinfhat}[1]{\myhat{#1}}
\newcommand{\ehat}[3]{\stackrel{\triangle {#3}}{{#1}_{#2}}}
\newcommand{\mycapt}[1]{\overline{#1}^{\it}}
\newcommand{\mycapnt}[1]{\overline{#1}^{\int}}
\newcommand{\mycap}[1]{\overline{#1}}
\newcommand{\ecap}[2]{\stackrel{=}{#1}_{#2}}
\newcommand{\alienst}[1]{\cA_\it(#1)}
\newcommand{\aliensnt}[1]{\cA_\int(#1)}
\newcommand{\aliens}[1]{\cA(#1)}
\newcommand{\ealiens}[2]{\overline{\cA(#1)}_{#2}}

\newcommand{\interp}[1]{\mbox{$[\!\mid \!\! #1 \!\! \mid\!]$}}
\newcommand{\lp}[1]{#1_{peak}}
\newcommand{\lc}[1]{#1_{conv}}

\newenvironment{myproof}[1]
{{\color{darkgray}\sffamily\bfseries Proof.} #1}

\begin{abstract}
We investigate the new, Turing-complete class of \emph{layered}
systems, whose lefthand sides of rules can only be overlapped at a
multiset of disjoint or equal positions. Layered systems define a
natural notion of rank for terms: the maximal number of
non-overlapping redexes along a path from the root to a leaf.
Overlappings are allowed in finite or infinite trees. Rules may be
non-terminating, non-left-linear, or non-right-linear.  Using a novel
unification technique, \emph{cyclic unification}, we show that rank
non-increasing layered systems are confluent provided their cyclic
critical pairs have cyclic-joinable decreasing diagrams.
\end{abstract}

\hide{
I focus on the new conditions for the cyclic-joinability of the cyclic critical 
pairs and also the new proof of the main theorem (critical case) for the 
moment. I think it is a priority to make sure this part is correct. I have 
3 comments for this:

(i) in Def-34, I think the additional condition should be:
``If Var(\lin{l}|_p) \subset Var(\lin{l}), then there is no i-facing step in the 
diagram, and each rewrite step u ->^q_k v of the diagram should further 
satisfy Var(u|_q) \subseteq Var(g\eta_S) or k < i . ``
There are 3 differences with the current version in the paper:
(a) Var(\lin{l}|_p) \subset Var(\lin{l}) should be used to ensure that, in the 
first case of Case 3 of proof of TH-35, there exists no variable in l[]_p, 
such that the ranks of equalization steps would not be larger than m, 
and m = n + 1.
(b) The absence of i-facing step should be imposed. Since the whole 
right part of the diagram in the proof should be put into the bottom part 
of DD, if i-facing step exists and satisfies the variable condition, it may be 
strictly smaller than neither < n+1 , i > nor < m , j > .
(c) Var(g\eta_S) should be used instead of Var(l|_p\eta_S). I think this 
variable condition is to ensure the rank of the instantiated step is smaller 
than or equal to m. Notice that Var(g\eta_S) is not the same as Var(l|_p\eta_S), 
even though g\eta_S =^{cc}_{R_S\eta_S} l|_p\eta_S. 

(ii) In the first case of Case 3 of proof of TH-35, I think it is important to 
point out that m = n+1 .

(iii) In the second case of Case 3 of proof of TH-35, not all steps are strictly 
smaller than < n+1 , i >. The fact is, in the joinability diagram of the critical 
pair, 
- the steps which are < i : they are all strictly smaller than < n+1 , i >;
- the steps which are < j : if they satisfy the variable condition, they are 
strictly smaller than < m , j >; if they satisfy the k<i condition, they are 
strictly smaller than < n+1 , i >;
- the j-facing step : if it satisfies the variable condition, it is smaller than or 
equal to < m , j > ; if it satisfies the k<i condition, i.e. j<i, it is strictly smaller 
than < n+1 , i >.
As for the whole diagram, the whole right part is belonging to the bottom 
part of the DD, while the left part may have facing step. 
}

%\vspace{-1mm}
\section{Introduction}
\label{s:intro}
\hide{
A rewrite peak (resp., local peak) is a triple $\langle s,u,t\rangle$
such that $u$ rewrites (resp., in one step) to both $s$ and $t$.  A
rewrite peak is joinable if $s$ and $t$ rewrite to some common $v$.  A
rewrite system is confluent if all rewrite peaks can be
joined. Confluence is a major property of rewrite systems of all
kinds.
}

Confluence of terminating systems is well understood: it can be
reduced to the joinability of local peaks by Newman's lemma, and to
that of critical ones, obtained by unifying lefthand sides of rules at
subterms, by Knuth-Bendix-Huet's lemma.  Confluence can thus be decided
by inspecting all critical pairs, see for example~\cite{DBLP:books/el/leeuwen90/DershowitzJ90}.

%Despite many efforts
Many efforts
notwithstanding~\cite{DBLP:journals/jacm/Rosen73,DBLP:journals/jacm/Huet80,DBLP:conf/ctrs/ToyamaO94,MOOO97,DBLP:conf/rta/Okui98,DBLP:journals/tcs/ManoO01,DBLP:conf/rta/Oostrom08,gomi,DBLP:conf/icalp/JouannaudO09,DBLP:journals/ipl/SakaiO10,DBLP:journals/jar/HirokawaM11,DBLP:conf/rta/ZanklFM11,DBLP:conf/lpar/KleinH12,SOO15,DBLP:conf/birthday/LiuJ14,DBLP:conf/rta/AotoTU14},
confluence of non-terminating systems is far from being understood in
terms of critical pairs. Only recently did this question make
important progress with van Oostrom's complete method for checking
confluence based on decreasing diagrams, a generalization of
joinability~\cite{DBLP:journals/tcs/Oostrom94,DBLP:conf/rta/Oostrom08}. In
particular, while Huet's result stated that linear systems are
confluent provided their critical pairs are strongly
confluent~\cite{DBLP:journals/jacm/Huet80}, Felgenhauer showed that
right-linearity could be removed provided parallel critical pairs have
decreasing diagrams~\cite{Felgenhauer13}. Knuth-Bendix's and
Felgenhauer's theorems can join forces in presence of both terminating
and non-terminating rules~\cite{DBLP:conf/rta/LiuDJ14}.
%jpj{Incorporated is imprecise enough, I think. Feel free to find a better formulation.}

We show here that rank non-increasing layered systems are confluent
provided their critical pairs have decreasing diagrams.  Our
confluence result for non-terminating non-linear systems by critical
pair analysis is the first we know of. Further, our result holds in
case critical pairs become infinite, solving a
long standing problem raised in~\cite{DBLP:journals/jacm/Huet80}.
Prior solutions to the problem existed under different assumptions
that could be easily
challenged~\cite{DBLP:conf/ctrs/ToyamaO94,gomi,DBLP:conf/lpar/KleinH12}.

Our results use a simplified version of sub-rewriting introduced
in~\cite{DBLP:conf/rta/LiuDJ14}, and a simple, but essential
revisitation of unification in case overlaps generate occur-check
equations: \emph{cyclic unification} is based on a new,
important notion of cyclic unifiers, which enjoy all good properties
of unifiers over finite trees such as existence of most general cyclic
unifiers, and can therefore represent solutions of occur-check
equations by simple rewriting means.

Terms are introduced in Section~\ref{s:terms}, labelled rewriting and
decreasing diagrams in Section~\ref{s:labelledr}, sub-rewriting in
Section~\ref{s:sub}, cyclic unification in Section~\ref{s:unif} and
layered systems in Section~\ref{s:layer} where our main result is
developed, before concluding in Section~\ref{s:conclusion}.

\section{Terms, substitutions, and rewriting}
\label{s:terms}

Given a \emph{signature} $\cF$ of \emph{function symbols} and a
denumerable set $\cX$ of \emph{variables}, $\TFX$ denotes the set of
finite or infinite rational \emph{terms} built up from $\cF$ and
$\cX$. We reserve letters $x,y,z$ for variables, $f,g,h$ for function
symbols, and $s,t,u,v,w$ for terms.  Terms are recognized by top-down
tree automata in which some $\omega$-states, and only those, are
possibly traversed infinitely many times. Terms are identified with
labelled trees. See~\cite{tata2007} for details.

\emph{Positions} are finite strings of positive integers. We use
$o,p,q$ for arbitrary positions, the empty string $\rootp$ for the
root position, and ``$\cdot$'' for concatenation of positions or sets
thereof. We use $\FPos{t}$ for the (possibly infinite) set of
non-variable positions of $t$, $t(p)$ for the function symbol at
position $p$ in $t$, $t|_p$ for the \emph{subterm} of $t$ at position
$p$, and $t[u]_p$ for the result of replacing $t|_p$ with $u$ at
position $p$ in $t$.  We may omit the position $p$, writing $t[u]$ for
simplicity and calling $t[\cdot]$ a \emph{context}. We use $\geordp$
for the partial prefix order on positions (further from the root is bigger),
$p\# q$ for incomparable positions $p,q$, called \emph{disjoint}.  
%We use $t[u]_P$ for the simultaneous replacement at a set $P$ of disjoint positions in $t$ by $u$. 
The order on positions is extended to finite
sets as follows: $P\geordp Q$ (resp.\ $P\gtordp Q$) if $(\forall p\in
P)(\exists q\in max(Q))\, p\geordp q$ (resp.\ $p\gtordp q$), where
$max(P)$ is the set of maximal positions in $P$.  We use $p$ for the
%singleton 
set $\{p\}$.
%We write $u[v_1,\ldots, v_n]_Q$ for $u[v_1]_{q_1}\ldots[v_n]_{q_n}$ if $Q=\{q_i\}_1^n$. 

We use $\Var{t_1,\ldots,t_n}$ for the set of variables occurring in
$\{t_i\}_i$. A term $t$ is \emph{ground} if $\Var{t}=\emptyset$, 
\emph{linear} if no variable occurs more than once in $t$. Given a
term $t$, we denote by $\lin{t}$ any linear term obtained by renaming,
for each variable $x\!\in\!\Var{t}$, the occurrences of $x$ at positions
$\{p_i\}_i$ in $t$ by \emph{linearized variable} $x^{k_i}$ such that
$i\!\neq\! j$ implies $x^{k_i}\!\neq\! x^{k_j}$. Note that
$\Var{\lin{s}}\cap\Var{\lin{t}}=\emptyset$ iff
$\Var{s}\cap\Var{t}=\emptyset$. Identifying $x^{k_0}$ with
$x$, $\lin{t}=t$ for a linear term $t$.

A \emph{substitution} $\sigma$ is an endomorphism from terms to terms
defined by its value on its \emph{domain} $\Dom{\sigma}:=\{x \,:\,
\sigma(x)\neq x\}$. Its \emph{range} is
$\Ran{\sigma}:=\bigcup_{x\in\Dom{\sigma}} \Var{x\sigma}$. We use
$\sigma_{|V}$ for the restriction of $\sigma$ to
$V\subseteq\Dom{\sigma}$, and $\sigma_{|\lnot X}$ for the restriction
of $\sigma$ to $\Dom{\sigma}\setminus X$. The substitution $\sigma$ is
said to be \emph{finite} (resp., a \emph{variable substitution}) if
for each $x\in\Dom{\sigma}$, $\sigma(x)$ is a finite term (resp., a
variable).  Variable substitutions are called \emph{renamings} when
also bijective.  A substitution $\gamma$ is \emph{ground} if for each
$x\in\cX$, $\gamma(x)$ is ground. We
use Greek letters for substitutions and postfix notation for their
application. 

The strict \emph{subsumption order} $\gtsubs$ on finite terms
(resp.\ substitutions) associated with the quasi-order $s \gesubs t$
(resp.\ $\sigma \gesubs \tau$) iff $s\!=\!t\theta$
(resp.\ $\sigma\!=\!\tau\theta$) for some substitution $\theta$, is
well-founded.

%We now recall the necessary background of term rewriting.  
A \emph{rewrite rule} is a pair of finite terms, written
${\displaystyle l\ra r}$, whose \emph{lefthand side} $l$ is not a
variable and whose \emph{righthand side} $r$ satisfies
$\Var{r}\subseteq\Var{l}$. A \emph{rewrite system} $R$ is a set of
rewrite rules. A rewrite system $R$ is \emph{left-linear}
(resp.\ \emph{right-linear}, \emph{linear}) if for every rule $l\ra
r\in R$, $l$ is a linear term (resp.\ $r$ is a linear term, $l$ and
$r$ are linear terms).  Given a rewrite system $R$, a term $u$
\emph{rewrites} to a term $v$ \emph{at} a position $p$, written
${\displaystyle u\dlrps{p}{R} v}$, if $u|_p=l\sigma$ and
$v=u[r\sigma]_p$ for some rule $l\ra r\in R$ and substitution
$\sigma$.  The term $l\sigma$ is a \emph{redex} and $r\sigma$ its
\emph{reduct}.  We may omit $R$ as well as $p$, and also replace the
former by the rule which is used and the latter by a property it
satisfies, writing for example $u \dlrps{\ordp P}{l\ra r} v$.
Rewriting \emph{terminates} if there exists no infinite rewriting
sequence issuing from some term. Rewriting is sometimes called
\emph{plain rewriting}.

Consider a \emph{local peak} made of two rewrites issuing from the
same term $u$, say $u\lrps{p}{l\ra r} v$
and $u \lrps{q}{g\ra d} w$. Following
Huet~\cite{DBLP:journals/jacm/Huet80}, we distinguish three cases:

$p\#q \mbox{ (disjoint case), } q\gtordp p\cdot \FPos{l}
\mbox{ (ancestor case), and }
q\in p\cdot \FPos{l} \mbox{ (critical case).}$

\noindent
Given two, possibly different rules $l\ra r, g\ra d$ and a position
$p\in\FPos{l}$ such that $\Var{l}\cap\Var{g}=\emptyset$ and $\sigma$
is a most general unifier of the equation $l|_p=g$, then $l\sigma$ is
the \emph{overlap} and $\langle r\sigma,\, l\sigma[d\sigma]_p\rangle$
the \emph{critical pair} of $g\ra d$ on $l\ra r$ at $p$.

Rewriting extends naturally to
lists of terms of the same length, hence to substitutions of the same
domain. See~\cite{DBLP:books/el/leeuwen90/DershowitzJ90,terese} for surveys.

\section{Labelled rewriting and decreasing diagrams}
\label{s:labelledr}
Our goal is to reduce confluence of a non-terminating rewrite system
$R$ to that of finitely many critical pairs. Huet's analysis of linear
non-terminating systems was based on Hindley's lemma, stating that a
non-terminating relation is confluent provided its local peaks are
joinable in at most one step from each
side~\cite{DBLP:journals/jacm/Huet80}. The more general analyses needed
here have been made possible by van Oostrom's notion of decreasing
diagrams for labelled relations.

\begin{definition}
 A \emph{labelled rewrite relation} is a pair made of a rewrite
  relation $\lrps{}{}$ and a mapping from rewrite steps
   to a set of labels $\cL$ equipped with a partial
  quasi-order $\geordl$ whose strict part $\ordl$ is well-founded.
  We write ${\displaystyle u\lrps{p,m}{R}v}$ for a rewrite
  step from $u$ to $v$ at position $p$ with label $m$ and rewrite
  system $R$. Indexes  $p,m,R$ may be omitted. 
  We also write $\alpha\ordl l$ (resp. $l\ordl \alpha$) if $m\ordl l$ (resp.
  $l\ordl m$) for all $m$ in a multiset $\alpha$.
\end{definition}

Given an arbitrary labelled rewrite step $\lrps{l}{}$, we
denote its projection on terms by $\lrps{}{}$, its inverse by
$\rlps{l}{}$, its reflexive closure by $\lrpseq{l}$, its symmetric
closure by ${\displaystyle \eqps{l}{}}$, its reflexive, transitive
closure by $\lrder{\alpha}$ for some word $\alpha$ on the alphabet of labels,
and its reflexive, symmetric, transitive closure, called
\emph{conversion}, by ${\displaystyle \dconv{\alpha}}$. We may
consider the word $\alpha$ as a multiset.
\hide{Given $u$, $\{v \,|\,
u\lrder{}{} v\}$ is the set of reducts of $u$. We say that a
reduct of $u$ is \emph{reachable} from $u$.}

The triple $v,u,w$ is said to be a \emph{local peak} if
$v\rlps{l}{}u\lrps{m}{}w$, a \emph{peak} if
$v\rlder{\alpha}u\lrder{\beta}w$, a \emph{joinability diagram} if
$v\lrder{\alpha}u\rlder{\beta}w$. The local peak $v\rlps{p,m}{l\ra
  r}u\lrps{q,n}{g\ra d}w$ is a \emph{disjoint, critical, ancestor}
peak if $p\# q, q\in p\cdot\FPos{l}, q\gtordp p\cdot\FPos{l}$,
respectively. The pair $v,w$ is \emph{convertible} if ${\displaystyle
  v\dconv{\alpha}w}$, \emph{divergent} if
$v\rlder{\alpha}u\lrder{\beta}w$ for some $u$, and \emph{joinable} if
$v\lrder{\alpha}t\rlder{\beta}w$ for some $t$.  The relation
$\lrps{}{}$ is \emph{locally confluent} (resp.\ \emph{confluent},
\emph{Church-Rosser}) if every local peak (resp.\ divergent pair,
convertible pair) is joinable.

Given a labelled rewrite relation $\lrps{l}{}$ on terms, we consider
specific conversions associated with a given local peak
called \emph{local diagrams} and recall the important subclass of van
Oostrom's decreasing diagrams and their main property: a relation all
whose local diagrams are decreasing enjoys the Church-Rosser property,
hence confluence. Decreasing diagrams were introduced
in~\cite{DBLP:journals/tcs/Oostrom94}, where it is shown that they
imply confluence,  and further developed in~\cite{DBLP:conf/rta/Oostrom08}.
The first version suffices for our needs.

\begin{definition}[Local diagrams]
  A \emph{local diagram} $D$ is a pair made
  of a \emph{local peak} $\lp{D}=~ v\,\rlps{}{} u \lrps{}{}w$ and a
  \emph{conversion} $\lc{D}=~ v\dconv{}w$. We call \emph{diagram
    rewriting} the rewrite relation $\Lrps{}{\cD}$ on 
  conversions associated with a set $\cD$ of local diagrams, in which
  a local peak is replaced by one of its associated conversions:
%\vspace{-2mm}
\[P~\lp{D}~Q ~\Lrps{}{\cD}~ P~\lc{D}~Q \mbox{ for some } D\in\cD\]
\end{definition}

\begin{definition}[Decreasing diagrams~\cite{DBLP:journals/tcs/Oostrom94}] 
  A local diagram $D$ with peak $v\rlps{l}{} u \lrps{m}{}w$ is
  \emph{decreasing} if its conversion $\lc{D}=v\dlrder{\alpha}{} s
  \lrpseq{m}{} s' \dlrder{\delta}{} \drlder{\delta'}{}
  t'\rlpseq{l}{}t\drlder{\beta}{}w$ satisfies the following
  \emph{decreasingness condition}: labels in $\alpha$ (resp.\ $\beta$)
  are strictly smaller than $l$ (resp.\ $m$), and labels in
  $\delta,\delta'$ are strictly smaller than $l$ or $m$.  The rewrites
  $s \lrpseq{m}{} s'$ and $t'\rlpseq{l}{}t$ are called the
  \emph{facing steps} of the diagram.
\hide{The rewrites $v\dlrder{\alpha}{} s$ and $t\drlder{\beta}{}w$, $s
  \dlrpseq{m}{} s'$ and $t'\drlpseq{l}{}t$, and $s' \dlrder{\delta}{}
  \drlder{\delta'}{} t'$ are called the \emph{side}, \emph{facing},
  and \emph{middle steps} of the diagram, respectively. A decreasing
  diagram $D$ is \emph{stable} if $C[D\gamma]$ is decreasing for
  arbitrary context $C[\cdot]$ and substitution $\gamma$.}
\end{definition}
\begin{theorem}[\!\!\cite{DBLP:conf/icalp/JouannaudO09}]
\label{t:proof-termination}
The relation $\Lrps{}{\cD}$ terminates for any set $\cD$ of decreasing
diagrams.
\end{theorem}
\begin{corollary}
\label{c:confluence}
Assume that $T\subseteq \TFX$ and $\cD$ is a set of decreasing
diagrams in $T$ such that the set of $T$-conversions is closed under
$\Lrps{}{\cD}$.  Then, the restriction of $\lrps{}{}$ to $T$ is
Church-Rosser if every local peak in $T$ has a decreasing diagram in
$\cD$.
\end{corollary}

%\vspace{-1mm} 
This simple corollary of Theorem~\ref{t:proof-termination} is a
reformulation of van Oostrom's decreasing diagram theorem which is
convenient for our purpose.

\section{Sub-rewriting}
\label{s:sub}
Consider the following famous system inspired by an abstract example
of Newman, algebraized by Klop and publicized by
Huet~\cite{DBLP:journals/jacm/Huet80}, $\mbox{NKH} = \{f(x,x) \ra
a,\; f(x,c(x))\ra b,\; g\ra c(g)\}$.  NKH is overlap-free, hence
locally confluent by Huet's
lemma~\cite{DBLP:journals/jacm/Huet80}. However, it enjoys
non-joinable non-local peaks such as $a \drlps{}{} f(g,g) \dlrps{}{}
f(g, c(g)) \dlrps{}{} b$.

The main difficulty with NKH is that non-joinable peaks are
non-local. To restore the usual situation for which the confluence of
a relation can be characterized by the joinability of its local peaks,
we need another rewrite relation whose local peaks capture the
non-confluence of NKH as well as the confluence of its confluent
variations.  A major insight of~\cite{DBLP:conf/rta/LiuDJ14} is that
this can be achieved by the \emph{sub-rewriting} relation, that allows
us to rewrite $f(g,c(g))$ in one step to either $a$ or $b$, therefore
exhibiting the pair $\langle a,b\rangle$ as a sub-rewriting critical
pair.  Sub-rewriting is made of a preparatory \emph{equalization}
phase in which the variable instances of the lefthand side $l$ of some
rule $l\ra r$ are joined, taking place before the rule is applied in
the \emph{firing phase}.  In~\cite{DBLP:conf/rta/LiuDJ14},
sub-rewriting required a signature split to define layers in terms,
the preparatory phase taking place in the lower layers. No a-priori
layering is needed here:

\begin{definition}[Sub-rewriting]
\label{d:modulo}
A term $u$ \emph{sub-rewrites} to a term $v$ at $p\in \Pos{u}$ for
some rule $l\ra r\!\in\! R$, written $u \lrps{p}{\Rsub} v$, if 
$u \,\lrdertwo{(\gtordp p\cdot\FPos{l})}{\R} \, u[l\theta]_p \lrps{p}{R}
u[r\theta]_p=v$ for some substitution $\theta$. 
The term $u|_p$ is called a \emph{sub-rewriting redex}.
\end{definition}

This definition of sub-rewriting allows \emph{arbitrary rewriting}
below the lefthand side of the rule until a redex is obtained. This
is the major idea of sub-rewriting, ensuring that $R\subseteq \Rsub
\subseteq R^*$. A simple, important property of a sub-rewriting redex
is that it is an instance of a linearized lefthand side of rule:

\begin{lemma}[Sub-rewriting redex]
\label{l:linredex}
Assume $u$ sub-rewrites to $u[r\sigma]_p$ with $l\ra r$ at position $p$.\\
% \,(\lrps{\gtordp p\cdot\FPos{l}}{\R})^* \, u[l\theta]_p \lrps{p}{R} u[r\theta]_p$.  
Then $u|_p=\lin{l}\theta$ for some $\theta$ s.t. $(\forall
x\in\Var{l})\,(\forall p_i\in\Pos{l}\,\mbox{s.t.}\, l(p_i)\!=\!x)\,
\theta(x^{p_i}) \lrdertwo{}{R} \sigma(x)$.\\ We say that $\sigma$ is
an \emph{equalizer} of $l$, and the rewrite steps from $\lin{l}\theta$
to $l\sigma$ are an \emph{equalization}.
\end{lemma}

% $u|_p$ is a \emph{linearized  redex},

\hide{
\begin{proof}
From the definitions of sub-rewriting and linearization.
\end{proof}
}

Sub-rewriting differs from rewriting modulo by being directional. It
differs from Klop's higher-order rewriting modulo
developments~\cite{terese} used by Okui for first-order
computations~\cite{DBLP:conf/rta/Okui98}, in that the preparatory
phase uses arbitrary rewriting. Having non-left-linear rules with
critical pairs at subterms seems incompatible with using developments.
Sub-rewriting differs as well from relative
rewriting~\cite{DBLP:journals/jar/HirokawaM11} in that the preparatory
phase must take place below variables. The latter condition is
essential to obtain plain critical pairs based on plain unification.

Assuming that local sub-rewriting peaks characterize the confluence of
NKH, we need to compute the corresponding critical pairs. Unifying
the lefthand sides $f(x,x)$ and $f(y,c(y))$ results in the conjunction
$x=y\land y=c(y)$ containing the \emph{occur-check equation} $y=c(y)$,
which prevents unification from succeeding on finite trees but allows it to
succeed on infinite rational trees: the critical peak has therefore an
infinite overlap $f(c^\omega, c^\omega)$ and a finite critical pair
$\langle a, b\rangle$. At the level of infinite trees, we then have an
infinite local rewriting peak $a\drlps{}{} f(c^\omega, c^\omega) =
f(c^\omega, c(c^\omega)) \dlrps{}{} b$, the properties of infinite
trees making the sub-rewriting preparatory phase
useless. Sub-rewriting therefore captures on finite trees some
properties of rewriting on infinite trees, here the existence of a
local peak. Computing the critical pairs of the sub-rewriting
relation is therefore related to unification over finite trees
resulting possibly in solutions over infinite rational trees. In the
next section, we develop a novel view of unification that will allow
us to capture both finite and infinite overlaps by finite means.

\section{Cyclic unification}
\label{s:unif}
This section is adapted
from~\cite{DBLP:conf/fgcs/Colmerauer84,DBLP:books/el/leeuwen90/DershowitzJ90,DBLP:conf/birthday/JouannaudK91}
by treating finite and infinite unifiers uniformly: equality of terms
is interpreted over the set of infinite rational terms \emph{when
  needed}.

\hide{
We will also use the novel notion of
\emph{cyclic unifiers} for encoding infinite ones via the congruence
closure generated by the occur-check equations.
}

An \emph{equation} is an oriented pair of finite terms, written $u=v$.
A \emph{unification problem} $P$ is a (finite) conjunction $\land_i\,
u_i=v_i$ of equations, sometimes seen as a multiset of pairs written
$\vect{u}=\vect{v}$. A \emph{unifier} (resp. a \emph{solution}) of $P$
is a substitution (resp., a ground substitution) $\theta$ such that
$(\forall i)\, u_i\theta=v_i\theta$. 
% written $\vect{u}\unifsub{\theta}\vect{v}$.  
A unifier describes a generally infinite set of solutions via its
ground instances.  A major usual assumption, ensuring that solutions
exist when unifiers do, is that the set $\cT(\cF)$ of ground terms is
non-empty.  A unification problem $P$ has a \emph{most general finite
  unifier} $mgu(P)$, whenever a finite solution exists, which is
minimal with respect to subsumption and unique up to variable
renaming.  Computing $mgu(P)$ can be done by the unifier-preserving
transformations of Figure~\ref{fig:unif}, starting with $P$ until a
\emph{solved form} is obtained, $\perp$ denoting the absence of
solution, whether finite or infinite. Our notion of solved form
therefore allows for infinite unifiers (and solutions) as well as
finite ones:

\begin{definition}%[\!\!\cite{DBLP:books/el/leeuwen90/DershowitzJ90,DBLP:conf/birthday/JouannaudK91}]
\label{d:sf}
\emph{Solved forms} of a unification problem $P$ different from $\perp$
are unification problems
$S= \vect{x}=\vect{u}\land\vect{y}=\vect{v}$
such that

(i) $\cP=\Var{P}\setminus (\vect{x}\cup\vect{y})$ is the
set of \emph{parameters} of $S$;

(ii) variables in $\vect{x}\cup\vect{y}$ (i.e. variables at lefthand
sides of equations) are all distinct;

(iii) $(\forall\, x=u\!\in\!\vect{x}=\vect{u})$,
$\Var{u}\subseteq\cP$;
%and ($u\in\cX$  implies $x,u \!\not\in\! \Var{S\!\setminus\!\{x \!=\! u\}})$

(iv) $(\forall\, y=v\!\in\!\vect{y}=\vect{v})$, $\Var{v}\subseteq
\cP\cup\vect{y}$, $\Var{v}\cap\vect{y}\neq\emptyset$ and
$v\not\in\cX$.
\end{definition}

Equations $y=v\in\vect{y}=\vect{v}$ are called \emph{cyclic} (or
\emph{occur-check}, the vocabulary originating
from~\cite{DBLP:conf/fgcs/Colmerauer84} used so far), $\vect{x}$ is
the set of \emph{finite} variables, and $\vect{y}$ is the set of
(infinite) \emph{cyclic} (or \emph{occur-check}) variables. A solved
form is a set of equations since $\vect{x}\cup\vect{y}$ is itself a
set and an equation $x=y$ between variables can only relate a finite
variable $x$ with a parameter $y$.

%Unification Rules
\newcommand{\unifrule}[4]{
\mbox{\bf #1} & #2 
& \!\!\!\longrightarrow & \!\!
#3 &\!\!\! #4\\
}
\newcommand{\Remove}{{\bf Remove}\xspace}
\newcommand{\Decompose}{{\bf Decomp}\xspace}
\newcommand{\Conflict}{{\bf Conflict}\xspace}
\newcommand{\Choose}{{\bf Choose}\xspace}
\newcommand{\Coalesce}{{\bf Coalesce}\xspace}
\newcommand{\Swap}{{\bf Swap}\xspace}
\newcommand{\Merge}{{\bf Merge}\xspace}
\newcommand{\Replace}{{\bf Replace}\xspace}
\newcommand{\Merep}{{\bf Merep}\xspace}
\newcommand{\Eliminate}{{\bf Eliminate}\xspace}

\begin{figure}[h]
\hrule
\[\begin{array}{llcll}
\unifrule{Remove}{s=s \land P}{P}{ }

\unifrule{Decomp}{f(\vect{s}) = f(\vect{t}) \land P}{\vect{s}=\vect{t}
  \land P }{}

\unifrule{Conflict}{f(\vect{s}) = g(\vect{t})\land P}{\perp}{\mbox{\bf if } f\neq g}

\unifrule{Choose}{y = x \land P}{x = y \land P}{\mbox{\bf if } x\not\in Var(P),\,y\in\Var{P}}

\unifrule{Coalesce}{x = y \land P}{x = y \land P\{x \mapsto y\}}
   {\mbox{\bf if } x,y\in Var(P),\,x\neq y}

\unifrule{Swap}{u = x \land P}{x = u \land P}{\mbox{\bf if } u\not\in \cX}

\unifrule{Merge}{x = s\land x=t \land P}{x=s\land s = t \land
  P}{\mbox{\bf if } x\in \cX,\, 0 < |s| \leq |t|}

\unifrule{Replace}{ x =s \land P}{x = s\land P\{x\mapsto s\}}
{\mbox{\bf if } x\in Var(P),\, x\not\in\Var{s},\, s\not\in\cX}

\unifrule{Merep}{y\!=\!x\land x\!=\!s\land P}{y=s\land x = s\land P}
{\mbox{\bf if } x\!\in\! Var(s),\, s\!\not\in\! \cX, y\!\not\in\!\Var{s,P}}
&&&&\hfill{\mbox{ and no other rule applies}}

\end{array}
\]
\vspace{-6mm}
\caption{Unification Rules}
\label{fig:unif}
\hrule
\end{figure}

\begin{example}[NKH] 
$f(x,x)=f(y,c(y)) \ra_{\Decompose} x=y\land x=c(y) \ra_{\Coalesce}
  x=y\land y=c(y) \ra_{\Merep} x=c(y)\land y=c(y)$. Alternatively, 
$f(x,x)=f(y,c(y)) \ra_{\Decompose} x=y\land x=c(y)
  \ra_{\Replace} c(y)=y\land x=c(y) \ra_{\Swap} y=c(y)\land x=c(y)$.
\end{example} 

\Choose and \Swap originate
from~\cite{DBLP:conf/fgcs/Colmerauer84}. \Replace and \Coalesce ensure
that finite variables (but parameters) do not occur in equations
constraining the infinite ones. \Merep is a sort of combination of
\Merge and \Replace ensuring condition $v\not\in\cX$ in
Definition~\ref{d:sf}, item (iv). Unification over finite trees has
another failure rule, called {\bf Occur-check}, fired in presence of
cyclic equations.

\begin{theorem}
Given an input unification problem $P$, the unification
rules terminate, fail if the input has no solution, and return a
solved form $S=\vect{x}=\vect{u}\land\vect{y}=\vect{v}$ otherwise.
\end{theorem}

\begin{proof}
Termination, characterization of solved forms, soundness,
are all adapted from~\cite{DBLP:conf/birthday/JouannaudK91}.

{\bf Termination.}
The quadruple $\langle nu, |P|, nvre, nvle \rangle$ is used to
interpret a unification problem $P$, where

- $nu$ is the number of unsolved variables (0 for $\perp$), where a
variable $x$ is \emph{solved} in $x=s\land P'$ if $x\not\in\Var{s,P'}$;

%(ii) $x\in\Var{s}$ and $(\forall u=v\in P)\, x\neq u\land x\neq v$,
\hide{
Note: need to replace by ``solvability'', that is distance to solved
form in terms of number of times the variable x left in an equation x=t.
}

- $|P|$ is the multiset ($\emptyset$ for $\perp$) of natural numbers
$\{max(|s|, |t|) : s=t\in P\}$\,;

- $nvle$ (resp. $nvre$) is the number of equations in $P$ whose
lefthand (resp., righthand) side is a variable and the other side is
not.

\Remove, \Decompose and \Conflict decrease $|P|$ without increasing
$nu$.  \Choose and \Coalesce both decrease $nu$. \Swap decreases
$nvre$ without increasing $nu$ and $|P|$.  \Merge decreases $nvle$
without increasing $nu$, $|P|$ and $nvre$. \Replace decreases $nu$.
%since $x$ becomes solved in all three cases and no solved variable
%becomes unsolved.
Now, when \Merep applies, no other rule can apply, and we can check
that no rules can apply either after \Merep (except another possible
application of \Merep). This can happen only finitely many times, by
simply reasoning on the number of equations whose both sides are
variables.

{\bf Solved form.}  We show by contradiction that the output $P$,
which is in normal form with respect to the unification rules, is a
solved form in case \Conflict never applies. First, $P$ must be a
conjunction of equations $x=s$, since otherwise \Decompose or \Swap
would apply. Let $\cP= \Var{P} \setminus (\vect{x}\cup\vect{y})$.
\vspace{-2mm}
\begin{itemize}
\item
Condition (i) is a definition.

\item
Condition (ii).  
Let $P= x=s\land x=t\land P'$. Either $s$ or $t$ is a variable, since
otherwise \Merge would apply.  Assume without loss of generality that
$s\in\cX$, call it $y$. If $x=y$, \Remove applies.  If
$y\not\in\Var{t,P'}$, then \Choose applies. Otherwise, \Coalesce
applies. Hence $\vect{x}, \vect{y}$ are all different \emph{sets}, and
$P$ is therefore itself a set.

Let now $\vect{x}=\vect{u}$ be a maximal (with respect to inclusion)
set of equations in $P$ such that $\Var{\vect{u}}\subseteq \cP$, and
$\vect{y}=\vect{v}$ be the remaining set of equations.

\item
Condition (iii). It is ensured by the definition of $\vect{x}=\vect{u}$.

\item
Condition (iv). Let $y=v\in\vect{y}=\vect{v}$.

Let now $x=u\in\vect{x}=\vect{u}$, hence $x\not\in\Var{u}$.  Assume
$x\in\Var{v}$.  If $u\not\in\cX$, then \Replace applies.  Otherwise,
if $u$ has no other occurrence in $P$, then \Choose applies, else
\Coalesce applies.  Therefore $\Var{v}\cap\vect{x}=\emptyset$ by
contradiction.

Assume $\Var{v}\cap\vect{y}=\emptyset$. Then $\Var{v}\subseteq\cP$,
which contradicts the maximality of $\vect{x}=\vect{u}$.

We are left to show that $v$ is not a variable. If it were, then
$v\in\vect{y}$.  First, $v\neq y$, otherwise \Remove applies. Let $P
= (y=v)\land P'$ with $v\in\vect{y}\setminus\{y\}$. Let $v=z$, there
must exist $(z=w)\in P'$ for some $w$, otherwise $z\in\cP$. Hence
$P' = (z=w)\land P''$.  Now, $y\not\in\Var{w,P''}$, otherwise
\Coalesce applies. Then we show $z\in\Var{w}$:
firstly, $w\not = z$, otherwise \Remove applies;
secondly, $w$ is not a variable, otherwise $w\not\in\Var{y,P''}$ lets \Choose
apply, while $w\in\Var{y,P''}$ makes \Coalesce available;
then if $z\not\in\Var{w}$, \Replace applies. Thus $z\in\Var{w}$, allowing 
\Merep, which contradicts that we have indeed a solved form.
\end{itemize}

{\bf Soundness.} The set of solutions is an invariant of the
unification rules. This is trivial for all rules but \Coalesce,
\Merge, \Replace, \Merep, for which it follows from the fact that
substitutions are homomorphisms and equality is a congruence.
\end{proof}

The solved form is a \emph{tree solved form} if $\vect{y}=\emptyset$,
and otherwise an \emph{$\Omega$ solved form} whose solutions are
infinite substitutions taking their values in the set of infinite
(rational) terms. 
We shall now develop our notion of \emph{cyclic unifier} capturing
both solved forms
%these two cases 
by describing the infinite unifiers of a problem $P$ as a pair of a
finite unifier $\sigma$ and a set of cyclic equations $E$ constraining
those variables that require infinite solutions. In case
$E=\emptyset$, then $P$ is a tree solved form and $\sigma=mgu(P)$. To
avoid manipulating infinite unifiers when $E\neq\emptyset$, we shall
work with the cyclic equations themselves considered as a ground
rewrite system.

\newcommand{\cc}[1]{=^{cc}_{#1}}
\newcommand{\ccc}[1]{\equiv^{cc}_{#1}}
\newcommand{\constant}{\diamond}

\begin{definition}[\!\!\cite{DBLP:journals/jacm/NelsonO80}]
\label{d:cc}
Given a set of equations $E$, we denote by $\cc{E}$ the equational
theory in which the variables in $\Var{E}$ are treated as constants,
also called \emph{congruence closure} $E$.
\end{definition}

We are interested in the congruence closure defined by cyclic
equations, seen here as a set $R$ of ground rewrite rules.  We
may sometimes consider $R$ as a set of equations, to be either solved
or used as axioms, depending on context.

\begin{definition}
\label{d:crs}
A \emph{cyclic rewrite system} is a set of rules
$R=\{\vect{y}\ra\vect{v}\}$ such that the unification problem
$\vect{y}=\vect{v}$ is its own solved form with $\vect{y}$ as the set
of infinite cyclic variables. Variables in $R$ are treated as constants.
\end{definition}

\begin{lemma}
\label{l:crs}
A cyclic rewrite system $R$ is ground and critical pair free, hence
Church-Rosser.
%, and its inverse $R^{-1}$ rewriting from right to left is terminating.
\end{lemma}

\hide{
\begin{proof}
Since a cyclic rewrite system is ground, its lefthand sides are
pairwise non-unifiable, hence $R$ is strongly confluent, hence
Church-Rosser since ground~\cite{DBLP:journals/jacm/Huet80}.
\end{proof}
}

We now introduce our definition of cyclic unifiers and solutions:

\begin{definition}
\label{d:cu}
A \emph{cyclic unifier} of a unification problem $P$ is a pair
$\langle\eta, R\rangle$ made of a substitution $\eta$ and a cyclic
rewrite system $R=\{\vect{y}\ra\vect{v}\}$, satisfying:

(i) $\Dom{\eta}\subseteq\Var{P}\setminus\vect{y}$,
$\Ran{\eta}\cap\vect{y} = \emptyset$, and
$\Ran{\eta}\cap\Dom{\eta}=\emptyset$\,;

(ii) $P$ and $P\land R$ have identical sets of solutions\,; and

(iii.a) $(\forall u=v\in P)\,u\eta\cc{R\eta}v\eta$, \,
or equivalently by Lemma~\ref{l:crs},

(iii.b) $(\forall u=v\in P)\, u\eta \lrdertwo{}{R\eta}\rldertwo{}{R\eta} v\eta$ .

\noindent
A \emph{cyclic solution} of $P$ is a pair $\langle\eta\rho,
R\rangle$ made of a cyclic unifier $\langle\eta, R\rangle$ of $P$
and an additional substitution $\rho$.
\end{definition}

We shall use (iii.a) or (iii.b) indifferently, depending on our needs,
by referring to (iii).

The idea of cyclic unifiers is that the need for infinite values for
some variables is encoded via the use of the cyclic rewrite system
$R$, which allows us to solve the various occur-check equations
generated when unifying $P$.  Finite variables are instantiated by the
finite substitution $\eta$, which ensures that cyclic unification
reduces to finite unification in the absence of infinite
variables. The technical restrictions on $\Dom{\eta}$ and $\Ran{\eta}$
aim at making $\eta$ idempotent. In (iii), parameters occurring in $R$
are instantiated by $\eta$ before rewriting takes place: cyclic
unification is nothing but rigid unification modulo the cyclic
equations in $R$~\cite{DBLP:conf/lics/GallierRS87}.  Instantiation of
the infinite variables $\vect{y}$ is delegated to cyclic solutions via
the additional substitution $\rho$ which may also instantiate the
variables introduced by $\eta$.

\begin{example}
\label{ex:cu}
Consider the equation $f(x,z,z) \!=\!f(a,y,c(y))$.  A cyclic unifier
is $\langle \{x\mapsto a\}, \{y \!\ra\! c(z), z\!\ra\!c(z)\}\rangle$,
and a cyclic solution is $\langle \{x \!\mapsto\! a, y\!\mapsto\!  a,
z\!\mapsto\!  c(a)\}, \{y \!\ra\! c(z), z\!\ra\!c(z)\}\rangle$, which
is clearly an instance of the former by the substitution $\{y\mapsto
a, z\mapsto c(a)\}$.  For the former, $ f(a,z,z)
\!\cc{\{y=c(z),z=c(z)\}}\! f(a,y,c(y))$.  Another cyclic unifier is
$\langle \{x \!\mapsto\! a\}, \{z \!\ra\! c(y), y \!\ra\!
c(y)\}\rangle$, for which $f(a,z,z)\!\cc{\{z=c(y),y=c(y)\}}\!f(a,
y,c(y)) $.
\end{example}

The set of cyclic unifiers of a problem $P$ is closed under
substitution instance, provided the variable conditions on its
substitution part are met, as is the set of its unifiers. Cyclic
unifiers have indeed many interesting properties similar to those of
finite unifiers, of which we are going to investigate only a few which
are relevant to the confluence of layered systems.

We now focus our attention on specific cyclic unifiers
sharing a same cyclic rewrite system.

\begin{definition}
\label{l:ecu}
Given a unification problem $P$
with solved form
$S=\vect{x}=\vect{u}\land\vect{y}=\vect{v}$, let

- its set of parameters
$\cP=\Var{P}\setminus(\vect{x}\cup\vect{y})$,

- its cyclic rewrite system $R_S=\{\vect{y}\ra\vect{v}\}$
and \emph{canonical} substitution $\eta_S=\{\vect{x}\mapsto\vect{u}\}$,

- its $S$-based cyclic unifiers $\langle \eta, R_S\rangle$, among
which $\langle\eta_S, R_S\rangle$ is
said to be \emph{canonical}.
\end{definition}

We now show a major property of $S$-based cyclic unifiers, true for any solved form $S$:

\begin{lemma}
\label{l:pocu}
Given a unification problem with solved form $S$, the set of $S$-based
cyclic unifiers is preserved by the unification rules.
\end{lemma}

\begin{proof}
The result is straightforward for \Remove, \Choose, and \Swap. It is
true for \Decompose and \Conflict since, using formulation (iii.b) of
Definition~\ref{d:cu}, the rules in $R\eta$ cannot apply at the root
of $\cF$-headed terms.  Next comes \Coalesce. We need to prove that
$\langle \eta,R\rangle$ is a cyclic unifier for $x=y\land P$ iff it is
one for $x=y\land P\{x\mapsto y\}$. Let $u=v\in P$. For the only if
case, we have $u\{x\mapsto y\}\eta=u\eta\{x\eta\mapsto
y\eta\}\cc{R\eta}u\eta \cc{R\eta}v\eta\cc{R\eta}v\eta\{x\eta\mapsto
y\eta\}=v\{x\mapsto y\}\eta$. The if case is similar.  \Replace is
similar to \Coalesce.  Consider now \Merge (\Merep is
similar). Showing that $\langle\eta,R\rangle$ is a cyclic unifier for
$x=s\land x=t \land P$ iff it is one for $x=s\land s=t \land P$ is
routine by using transitivity of the congruence closure $\cc{R\eta}$.
\end{proof}

We can now conclude:

\begin{theorem}
\label{t:mgcu}
Given a unification problem $P$ with solved form
$S\!=\!\vect{x}\!=\!\vect{u}\land\vect{y}\!=\!\vect{v}$, the
canonical $S$-based cyclic unifier is most general among the set of
$S$-based cyclic unifiers of $P$.
\end{theorem}

\begin{proof}
Let $\langle \eta, R_S\rangle$ be a cyclic unifier of $P$ based on
$S$.  

Let $x=u\in\vect{x}=\vect{u}$. By definition of cyclic
unification, $x\eta \lrdertwo{}{R_S\eta} \rldertwo{}{R_S\eta} u\eta$. By
definition of a solved form and cyclic unifiers, we have:
$\Var{x\eta,u\eta}\subseteq(\vect{x}\cup\cP\cup\Ran{\eta})$,
$(\vect{x}\cup\cP)\cap\vect{y}=\emptyset$,
$\Ran{\eta}\cap\vect{y}=\emptyset$, and
$\vect{y}\cap\Dom{\eta}=\emptyset$. Therefore, $x\eta$ and $u\eta$ are
irreducible by $R_S\eta$. Hence $x\eta=u\eta$. Since $x\eta_S=u$, it
follows that $x\eta=u\eta= (x\eta_S)\eta=x(\eta_S\eta)$.

Let now $z\in\Var{P}\setminus\vect{x}$. Since
$z\not\in\Dom{\eta_S}$, then
$\eta(z)=z\eta=(z\eta_S)\eta=z(\eta_S\eta)$.

Therefore, $\eta=\eta_S\eta$ and we are done.
\end{proof}

This result, which suffices for our needs, is easily lifted to cyclic
solutions, as they are instances of a cyclic unifier. We can further prove
that $\eta_S$ is more general than any $S'$-based cyclic unifiers, for
any solved form $S'$ of $P$. This is where our conditions on
$\Ran{\eta}$ become important. We conjecture that it is most general
among the set of all cyclic unifiers.

\section{Layered systems}
\label{s:layer}
\newcommand{\StOF}{\mbox{(DLO)}\xspace}
\newcommand{\SuOF}[1]{\mbox{SOF($#1$)\xspace}}
\newcommand{\OF}[1]{\mbox{OF($#1$)}\xspace}
\newcommand{\Conv}[2]{\mbox{Conv$_n^\theta$\!($#1$,$#2$)}\xspace}
\newcommand{\Equal}[1]{\mbox{Equalize$_n$\!($#1$)$^\theta_\sigma$}\xspace}

NKH is non-confluent, but can be easily made confluent by adding the
rule $a\ra b$ (giving NKH$^1$), or removing the rule $g\ra c(g)$
(giving NKH$^2$). It can be made non-right-ground by making the
symbols $a,b$ unary (using $a(c(x))$ and $b(c(x))$ in the righthand
sides of rules, giving NKH$^3$), or even non-right-linear by making
them binary (giving NKH$^4$). There are classes of systems containing
NKH for which it is possible to conclude its non-confluence. The
following classes succeed for NKH$^1$:
simple-right-linear~\cite{DBLP:conf/ctrs/ToyamaO94}, strongly
depth-preserving~\cite{gomi}, and relatively
terminating~\cite{DBLP:conf/lpar/KleinH12}.  As for NKH$^3$, it is
neither simple-right-linear nor strongly depth-preserving: only
\cite{DBLP:conf/lpar/KleinH12} can cover it. When it comes to NKH$^4$,
relative termination becomes hard to satisfy in presence of
non-right-linearity~\cite{DBLP:conf/lpar/KleinH12}.

Our goal is to define a robust, Turing-complete class of rewrite
systems capturing NKH and its variations, for which confluence can
be analyzed in terms of critical diagrams.
%\mo{Why robust? }
\begin{definition}
A rewrite system $R$ is \emph{layered} iff it satisfies the
\emph{disjointness} assumption \StOF that \emph{linearized}
%sub-rewriting 
overlaps of some lefthand sides of rules upon a given lefthand side
$l$ can only take place at a multiset of disjoint or equal positions
of $\FPos{l}$:
\[\begin{array}{ccl}
\StOF & := & (\forall l\ra r \in R)\, (\forall p\in\FPos{l})\,
(\forall g\ra d\in R \mbox{
  s.t. }\Var{\lin{l}}\cap\Var{\lin{g}}=\emptyset)\\
& & (\forall\sigma:\Var{\lin{l}|_p, \lin{g}} \ra \TFX \mbox{ s.t. }
\lin{l}|_p\sigma=\lin{g}\sigma)\, \SuOF{l|_p}\land\SuOF{g} \\
\SuOF{u} & := & (\forall q\in\FPos{u}\!\setminus\!\{\rootp\})\,
\OF{u|_q}\\
\OF{v} & := & (\forall g\ra d\!\in\! R \mbox{
  s.t. }\Var{\lin{v}}\cap\Var{\lin{g}}=\emptyset)\, \\
& & (\forall o\in\FPos{v})(\forall\sigma:\Var{\lin{v}, \lin{g}} \ra
\TFX)\,\, \lin{v}|_o\sigma\neq \lin{g}\sigma
\end{array}\]
\end{definition}
SOF stands for \emph{subterm overlap-free}, and OF for
\emph{overlap-free}. In words, if two lefthand sides of rules in $R$
overlap (linearly) a lefthand side $l$ of a rule in $R$ at positions
$p$ and $q$ respectively, then either $p=q$ or $p\#q$. Overlaps at
different positions along a path from the root to a leaf of $l$ are
forbidden.

\hide{
on a given path from the root to a leaf
of a lefthand side of rule $l$, there can be at most one position where
$l|_p$ is (linearly) unifiable with a lefthand side of rule in $R$.
Different overlaps may still take place at positions disjoint from $p$.

In other words, all linearized overlaps occurring
on a given path from the root to a leaf of a lefthand side must take
place at the same position. Overlaps may still take place at disjoint
positions on different paths. 
}

%The statement \SuOF{l|_p}$\land$\SuOF{g} happens to be redundant in case of root-overlaps.

Layered systems is a decidable class that relates to overlay
systems~\cite{DBLP:conf/ctrs/DershowitzOS87}, for which overlaps
computed with plain unification can only take place at the root of
terms --hence their name--, and generalizes strongly non-overlapping
systems~\cite{DBLP:journals/ipl/SakaiO10} which admit no linearized
overlaps at all. All these classes are Turing-complete since they
contain a complete class~\cite{Klop93}.

\begin{example}
NKH is a layered system, which is also overlay. $\{h(f(x,y)) \ra a,
f(x,c(x)) \ra b\}$ is layered but not overlay. $\{h(f(x,x)) \ra a,
f(x,c(x))\ra b, g\ra c(g)\}$ is layered, but not strongly
non-overlapping. $\{f(h(x))\ra x, h(a)\ra a, a\ra b\}$ is not overlay
nor layered: \SuOF{h(x)} succeeds while \SuOF{h(a)} fails, hence their
conjunction fails.
\end{example}

\subsection{Layering}
\label{ss:layering}
We define the rank of a term $t$ as the maximum number of
non-overlapping linearized redexes traversed from the root to some
leaf of $t$, which differs from the usual redex-depth.

\begin{definition}
\label{d:rank}
Given a layered rewrite system $R$, the \emph{rank} $\rank{t}$ of a
term $t$ is defined by induction on the size of terms as follows:

- the maximal rank of its immediate subterms if $t$ is not a linearized
redex\,; otherwise,

- 1 plus $max\{\rank{\sigma} : (\exists l\ra r\in
R)\,t=\lin{l}\sigma\}$, where $\rank{\sigma} := max\{\rank{\sigma(x)} :
x\in\Var{\lin{l}}\}$.
\end{definition}

\begin{definition}
A rewrite system $R$  is \emph{rank non-increasing} if
for all terms $u,v$ such that $u\dlrps{}{R}v$, then $\rank{u}\geq\rank{v}$.
\end{definition}

The rewrite system $\{f(x)\ra c(f(x))\}$ is rank non-increasing while
$\{f(x) \ra f(f(x))\}$ is rank increasing.  The system $\{fib(0)\ra 0,
fib(S(0))\ra S(0), fib(S(S(x))) \ra fib(S(x)) + fib(x)\}$ calculating
the Fibonacci function is rank non-increasing.  NKH is rank
non-increasing. The coming decidable sufficient condition for rank
non-increasingness captures our examples (but Fibonacci, for which an
even more complex decidable property is needed):

\begin{lemma}
\label{l:rkcond}
A layered rewrite system $R$ is rank non-increasing if each rule 
$g\ra d$ in $R$ satisfies the following properties:
\begin{enumerate}[(i)]
\item
$((\forall l\ra r\in R)(\forall l'\ra r'\in R)
\mbox{ s.t. }\Var{d},\Var{\lin{l}},\Var{\lin{l'}}\mbox{ are pairwise disjoint})\\ 
(\forall p,q\in\FPos{d}\mbox{ s.t. } q>p\cdot\FPos{l})\\
(\forall\sigma:\Var{g,\lin{l},\lin{l'}}\ra\GTF) 
\,\,(d|_p\sigma\not=\lin{l}\sigma) \lor (d|_q\sigma\not=\lin{l'}\sigma)$\,;

\item
$((\forall l\ra r\in R)
\mbox{ s.t. }\Var{g}\cap\Var{\lin{l}}=\emptyset)
(\forall p\in\FPos{l}\setminus\rootp)\\
(\forall\sigma:\Var{g,\lin{l}}\ra\GTF \mbox{ s.t. }
d\sigma=\lin{l}|_p\sigma) \\
((\exists l'\ra r'\in R)\mbox{ s.t. }\Var{\lin{l'}}\cap\Var{g,\lin{l}}=\emptyset)
(\exists x\!\not\in\!\Var{\lin{l},\lin{l'},g})
\,\,\lin{l}[x]_p \,\gesubs\, \lin{l'}$.
\end{enumerate}
\end{lemma}

We can now index term-related notions by the rank of terms.
Let \Tn (in short, $\cT_n$) be the set of terms of rank at most $n$.
Two terms in $\cT_n$ are \emph{$n$-convertible}
(resp. \emph{$n$-joinable}) if their \R-conversion
(resp. \R-joinability) involves terms in $\cT_n$ only.

\subsection{Closure properties}
\label{ss:}
\hide{
Call \emph{OF-term} a term whose all subterms $u$ satisfy \OF{u}, and
\emph{OF-substitution} a substitution mapping variables to OF-terms. }
Call a term $u$ an \emph{OF-term} if $u$ satisfies \OF{u}, and a
substitution an \emph{OF-substitution} if it maps variables to
OF-terms.  OF-terms enjoy several important closure properties. Given
two substitutions $\theta,\sigma$ and rank $n$, let

$\Conv{\lin{u}}{\lin{v}}$ iff $\lin{u}\theta$ and $\lin{v}\theta$ are
$n$-convertible, and 

$\Equal{\lin{u}}$ iff $\lin{u}\theta \lrdertwo{}{\Rsub} u\sigma$ with
terms of rank at most $n$.

\begin{lemma}
\label{l:sof}
For all OF-terms $u$ and substitutions $\gamma$,
$u\gamma$ cannot sub-rewrite at a position $p\!\in\!\FPos{u}$.
\end{lemma}

\begin{corollary}
\label{c:OFinst}
OF-terms are preserved under instantiation by OF-substitutions.
\end{corollary}

\begin{lemma}
\label{l:conv}
Let $u,v$ be two terms such that $\Conv{\lin{u}}{\lin{v}}$,
$\Equal{\lin{u}}$ and $\Equal{\lin{v}}$. Then $u\sigma$ and
$v\sigma$ are $n$-convertible.
\end{lemma}

\hide{
\begin{proof}
By building a conversion $u\sigma\rldertwo{}{\Rsub}
\lin{u}\theta \dconvtwo{}{\Rsub} \lin{v}\theta\lrdertwo{}{\Rsub} v\sigma$
satisfying the conclusion.
\end{proof}
}

\begin{lemma}
\label{l:unifinvar}
Let $\land_i\, u_i=v_i$ be obtained by decomposition of a unification
problem $P$.  Assume all equations $u_i=v_i$ satisfy the properties
$\Conv{\lin{u_i}}{\lin{v_i}}$, $\Equal{\lin{u_i}}$,
$\Equal{\lin{v_i}}$, $\OF{u_i}$ and $\OF{v_i}$.  Assume further that
$n$-convertible terms are joinable. Then, unification of $P$ succeeds,
and returns a solved form whose all equations satisfy these five
properties.
\end{lemma}

\noindent
In this lemma and coming proof, we assume that linearizations are
propagated by the unification rules, implying in particular that
$\lin{u|_p}=\lin{u}|_p$. $P$ defines the initial linearization.

\begin{myproof}
We show that these five properties are invariant by the unification
rules. The claim follows since the unification rules
terminate. We use notations of Figure~\ref{fig:unif}.
\begin{itemize}
\item
\Remove, \Choose, \Swap are straightforward.

\item
\Decompose. By assumption,
$\Conv{\lin{f(\vect{s})}}{\lin{f(\vect{t})}}$, hence
$\lin{f(\vect{s})}\theta$ and $\lin{f(\vect{t})}\theta$ are joinable by
using terms of rank at most $n$, since $R$ is rank non-increasing.  By
assumption \OF{f(\vect{s})} and \OF{f(\vect{t})}, hence no rewrite
can take place at the root. The result follows.

\item
\Conflict. By the same token, $f=g$, a contradiction.
Thus \Conflict is impossible.

\item
\Coalesce. By assumption, $\Conv{x^k}{y^l}$,
$\Equal{x^k}$, $\Equal{y^l}$, and $(\forall u\!=\!v\in
P)$, $\Conv{\lin{u}}{\lin{v}}$, $\OF{u}$, $\Equal{\lin{u}}$, $\OF{v}$ 
and $\Equal{\lin{v}}$.  Putting these things together,
we get $\Conv{\lin{u}\{x^k\mapsto y^l\}}{\lin{v}\{x^k\mapsto y^l\}}$,
hence $\Conv{\lin{u\{x\mapsto y\}}}{\lin{v\{x\mapsto y\}}}$.
Similarly, properties $\Equal{\lin{u\{x\mapsto y\}}}$ and
$\Equal{\lin{v\{x\mapsto y\}}}$ hold. Property $\OF{u}$ is of course
preserved by variable renaming for any $u$.

\item
\Merge. Assume $\Conv{x^k}{\lin{s}}$, $\Conv{x^l}{\lin{t}}$, $\OF{s}$,
$\Equal{\lin{s}}$, $\Equal{x^k}$, $\OF{t}$, $\Equal{{\lin{t}}}$ and 
$\Equal{x^l}$.  $\Conv{\lin{s}}{\lin{t}}$ follows from
$\Conv{x^k}{\lin{s}}$, $\Conv{x^l}{\lin{t}}$, $\Equal{x^k}$ and 
$\Equal{x^l}$. The other properties follow similarly.

\item
\Replace.
The proof is similar for the first 3 properties. Further, OF is
preserved by replacement by Corollary~\ref{c:OFinst}.

\item
\Merep. Similar to \Merge.
\qed
\end{itemize}
\end{myproof}

\begin{example} [NKH]
Let $P= f(x,x)=f(y,c(y))$. Then $P \ra_{\Decompose} x=y\land
x=c(y)\ra_{\Replace} c(y)=y\land x=c(y) \ra_{\Swap} y=c(y)\land
x=c(y)$. Successive linearizations yield $f(x^1,x^2) \!=\!
f(y^1,c(y^2))$, $x^1 \!=\! y^1\land x^2 \!=\! c(y^2)$, $c(y^2) \!=\!
y^1\land x^2 \!=\! c(y^2)$ and $y^1 \!=\! c(y^2)\land x^2 \!=\!
c(y^2)$. The announced properties of the solved form can be easily
verified.
\end{example}

\begin{corollary}
\label{c:cp}
Let $l\ra r, g\ra d \!\in\! R$ and $p\in\FPos{l}$ such that
$\Var{l}\!\cap\!\Var{g}\!=\!\emptyset$, and
$\lin{l}|_p\theta=\lin{g}\theta$ are terms in $\cT_{n+1}$. Then,
unification of $l|_p\!=\!g$ succeeds, returning a solved form $S$
s.t., for each $z=s\in S, \Conv{\lin{z}}{\lin{s}}$, $\OF{s}$,
$\Equal{\lin{s}}$ for all $\sigma$ satisfying
$(\lin{l}\theta\lrdertwo{(>\FPos{l})}{}l\sigma) \land
(\lin{g}\theta\lrdertwo{(>\FPos{g})}{}g\sigma)$, and further,
$\SuOF{l|_p\eta_S}\!\land\!\SuOF{g\eta_S}$.
\end{corollary}

\begin{proof}
Unification applies first \Decompose. Conclude by
Lemmas~\ref{l:unifinvar} and Corollary~\ref{c:OFinst}.
\end{proof}

\begin{corollary}
\label{c:rank}
Assume $t=\lin{l}\sigma$ for some $l\ra r\in R$. Then, $\rank{t} = 1
+\rank{\sigma}$.
\end{corollary}

\begin{proof}
Let $t=\lin{l_i}\sigma_i=\lin{l_i}\theta\gamma$ (note that $\gamma$
does not depend on $i$), where $\theta=mgu(=_{i}\lin{l_i})$. Then,
$\rank{t} = 1 + max_i\{\rank{\sigma_i}\}=1+
max_i\{\rank{\theta\gamma}\}= 1+ \rank{\gamma}= 1+\rank{\sigma_i}$
since $\theta$ satisfies OF at all non-variable positions by
Lemma~\ref{l:unifinvar}.
\end{proof}

\begin{example} [NKH]
Consider $f(c(g),c(g))$ of rank 2, using either linearized lefthand
side $f(x^1,x^2)$ or $f(y^1, c(y^2))$ to match
$f(c(g),c(g))$. Corresponding substitutions have rank 1.
\end{example}

A major consequence is that the preparatory phase of sub-rewriting
operates on terms of a strictly smaller rank. This would not be true
anymore, of course, with a conversion-based preparatory phase. More
generally, we can also show that the rank of terms does not increase
--but may remain stable-- when taking a subterm, a property which is
not true of non-layered systems. Consider the system $\{f(g(h(x))) \ra
x, g(x)\ra x, h(x)\ra x\}$. The redex $f(g(h(a)))$ has rank 1 with our
definition, but its subterm $g(h(a))$ has rank 2.

\subsection{Testing confluence of layered systems via their cyclic critical pairs}
\label{s:ctcp}

Since $R$ is rank non-increasing we shall prove confluence by
induction on the rank of terms.  Since rewriting is rank
non-increasing, the set of $\cT_n$-conversions is closed under diagram
rewriting, hence allowing us to use Corollary~\ref{c:confluence}.
This is why we adopted this restricted, but complete, form of
decreasing diagram rather than the more general form described
in~\cite{DBLP:conf/rta/Oostrom08}.

\begin{definition}[Cyclic critical pairs]
\label{d:ccp}
Given a layered rewrite system $R$, let $l \ra r, g\ra d\in R$ and
$p\in\FPos{l}$ such that $\Var{l}\cap\Var{g}=\emptyset$, and $l|_p=g$
is unifiable with canonical cyclic unifier $\langle
\eta_S=\{\vect{x}\mapsto\vect{u}\},
\R_S=\{\vect{y}\ra\vect{v}\}\rangle$.  Then, $r\eta_S \rlps{}{R}
l\eta_S \cc{R_S\eta_S} l[g]_p\eta_S \lrps{}{R} l[d]_p\eta_S$ is a
\emph{cyclic critical peak}, and $\langle r\eta_S,l[d]_p\eta_S
\rangle$ is a \emph{cyclic critical pair}, which is said to be
\emph{realizable} by the substitution $\theta$ iff $(\forall y\ra v\in
R_S)\, y\theta \lrdertwo{}{R} \,\rldertwo{}{R} v\theta$.
\end{definition}

The relationship between critical peaks and realizable cyclic critical
pairs, usually called critical pair lemma, is more complex than usual:

\begin{lemma}[Cyclic critical pair lemma]
\label{l:equalization}
\mbox{Let $l \lrps{}{} r, g\lrps{}{} d\!\in\! R$ such that
  $\Var{l}\cap\Var{g}\!=\!\emptyset$.}  Let $r\sigma
\rlps{\rootp}{l\ra r} l\sigma \rlder{(>\FPos{l})}{} \lin{l}\theta =
\lin{l}\theta[\lin{g}\theta]_p \lrder{(>p\cdot\FPos{g})}{}
\lin{l}\theta[g\sigma]_p \lrps{p}{g\ra d} \lin{l}\theta[d\sigma]_p$ be
a sub-rewriting local peak in $\cT_{n+1}$, satisfying $p\in\FPos{l}$
and $\Var{\lin{l}\theta}\cap\Var{l,g}=\emptyset$. Assume further that
$R$ is Church-Rosser on the set $\cT_n$. Then, there exists a cyclic
solution $\langle\gamma, R_S\rangle$ such that $S$ is a solved form of
the unification problem $l|_p=g$, $\gamma=\eta_S\rho$ for some $\rho$
of domain included in $\Var{l,g}$, $\sigma\lrdertwo{}{R}\gamma$, and
$R_S$ is realizable by $\gamma$.
\end{lemma}

\begin{proof}
Corollary~\ref{c:cp} asserts the existence of a solved form
$S=\vect{x}=\vect{u}\land\vect{y}=\vect{v}$ of the problem $l|_p=g$.
But $\langle\sigma,R_S\rangle$ may not be a cyclic solution.  We shall
therefore construct a new substitution $\gamma$ such that $\sigma
\lrdertwo{}{\Rsub} \gamma$ and $\langle\gamma,R_S\rangle$ is a cyclic
solution of the problem, obtained as an instance by some substitution
$\rho$ of the most general cyclic unifier $\langle\eta_S, R_S\rangle$
by Theorem~\ref{t:mgcu}.

The construction of $\gamma$ has two steps. The first aims at forcing
the equality constraints given by $S$. This step will result in each
parameter having possibly many different values. The role of the
second step will be to construct a single value for each parameter.
%Both steps are based on rewriting.
\hide{
We then
start equalizing independently the equations $x\!=\!s \!\in\!S$.
Since $\Conv{x^j}{\lin{s}}$, $\Equal{x^j}$ and $\Equal{\lin{s}}$,
$x\sigma$ and $s\sigma$ are $n$-convertible by Lemma~\ref{l:conv}. By
assumption, $x\sigma$ and $s\sigma$ are joinable, hence there exists a
term $t_x^s$ such that $x\sigma \lrdertwo{}{R} t_x^s \rldertwo{}{R}
s\sigma$. Since $\OF{s}$ by Corollary~\ref{c:cp}, the derivation
from $s\sigma$ to $t_x^s$ must occur at positions below
$\FPos{s}$. For $p\in\cP$, consider the set $G_p = \{ (t_x^s)|_q :
x=s\in S \mbox{ and } s|_q=p\}$, whose elements are terms rewritten
from some $p\theta$, hence are $n$-convertible to each other since
rewriting is rank non-increasing. By the Church-Rosser assumption,
they can all be rewritten to a same term $t_p$.  We now define
$\gamma$: 

(i) parameters. Given $p\in\cP$, we define $\gamma(p)\!=\!t_p$.  By
construction, $p\sigma\lrdertwo{}{R} t_p=p\gamma$.

(ii) finite variables. Given $x=u\in \vect{x}=\vect{u}$, let
$\gamma(x) = u\gamma_{|\cP}$, thus $x\gamma=u\gamma$. By construction,
$x\sigma \lrdertwo{}{R} u\tau_{x}^{u} \lrdertwo{}{R} u\gamma$, hence
$x\sigma\lrdertwo{}{R} x\gamma$.

(iii) cyclic variables.  Given $y=v\in \vect{x}=\vect{u}$, let
$\gamma(y) = v\gamma_{|\cP}$, making $y\sigma\lrdertwo{}{R} y\gamma$ trivial.

We proceed to show that $\gamma$ is a cyclic unifier of $S$, hence of
$l|_p=g$ by lemma~\ref{l:pocu}.  Let us first consider
$x=u\in\vect{x}=\vect{u}$. Since $x\not\in\Var{u}$, 
$x\gamma=u\gamma$.
%, hence $\gamma=\eta_S\rho$. 
By our assumptions, $\Var{r/d}\subseteq\Var{l/g}$,
$\Ran{\gamma}\subseteq\Var{u}$, $\Var{u}\cap\Var{l,g}=\emptyset$ and
$\Var{S} \subseteq\Var{l,g}$, hence $\gamma$ is an instance of
$\eta_S$ by $\gamma_{|\lnot\vect{x}}$. By definition, $\langle \eta_S
, \{\vect{y}\ra\vect{v}\}\rangle$ is an $S$-based cyclic unifier of
$S$.  Hence $\gamma$ is an $S$-based cyclic solution of $S$ satisfying
our statement.

We end up the proof by noting that $\gamma$ is a realizer of $R_S$.
}

We start equalizing independently equations $z\!=\!s\!\in\!S$.  Since
$\Equal{z^j}$, $\Equal{\lin{s}}$ and $\Conv{z^j}{\lin{s}}$, $z\sigma$
and $s\sigma$ are $n$-convertible by Lemma~\ref{l:conv}.  By
assumption, $z\sigma$ and $s\sigma$ are joinable, hence there exists a
term $t_z^s$ such that $z\sigma \lrdertwo{}{R} t_z^s \rldertwo{}{R}
s\sigma$. Since $\OF{s}$ by Corollary~\ref{c:cp}, the derivation from
$s\sigma$ to $t_z^s$ must occur at positions below
$\FPos{s}$. Maintaining equalities in $s\sigma$ between different
occurrences of each variable in $\Var{s}$, we get $t_z^s=s\tau_z^s$
for some $\tau_z^s$.  For each parameter $p$,
$p\sigma\lrdertwo{}{R}p\tau_z^s$, hence the elements of the non-empty
set $\{ p\tau_z^s : p\in\Var{s}\mbox{ for some } z=s\in S\}$ are
$n$-convertible thanks to rank non-increasingness. By our Church-Rosser
assumption, they can all be rewritten to a same term $t_p$.  We now
define $\gamma$:

(i) parameters. Given $p\in\cP$, we define $\gamma(p)\!=\!t_p$.  By
construction, $p\sigma\lrdertwo{}{R} t_p=p\gamma$.

(ii) finite variables. Given $x=u\in \vect{x}=\vect{u}$, let
$\gamma(x) = u\gamma_{|\cP}$, thus $x\gamma=u\gamma$. By construction,
$x\sigma \lrdertwo{}{R} u\tau_{x}^{u} \lrdertwo{}{R} u\gamma$, hence
$x\sigma\lrdertwo{}{R} x\gamma$.

(iii) cyclic variables.  Given $y=v\in \vect{x}=\vect{u}$, let
$\gamma(y) = y\sigma$, making $y\sigma\lrdertwo{}{R} y\gamma$ trivial.

(iv) variables in $\Var{l,g}\setminus\Var{l|_p,g}$, that is, those variables from
the context $l[\cdot]_p$ which do not belong to the unification problem 
$l|_p=g$, hence to the solved form $S$. Given
$z\in\Var{l,g}\setminus\Var{l|_p,g}$, let $\gamma(z)=z\sigma$, making
$z\sigma\lrdertwo{}{R} z\gamma$ trivial.

Therefore $\sigma\lrdertwo{}{R}\gamma$. We proceed to show
$\langle\gamma, R_S\rangle$ is a cyclic solution of $l|_p=g$. Take
$\rho=\gamma|_{\lnot\vect{x}}$. It is routine to see
$\gamma=\eta_S\rho$, and to check that $\langle\eta_S, R_S\rangle$ is
a cyclic unifier of $S$ by Definition~\ref{d:cu}, hence of $l|_p=g$ by
Lemma~\ref{l:pocu}. Hence the statement.

We end up the proof by noting that $\gamma$ is a realizer of $R_S$.
\end{proof}

In case of NKH, the lemma is straightforward since solved forms have no parameters.

%\subsection{A hierarchy of convertibility relations}
%\label{ss:strategy}

Our proof strategy for proving confluence of layered systems is as
follows: assuming that $n$-convertible terms are joinable, we show
that $(n+1)$-convertible terms are $(n+1)$-joinable by exhibiting
appropriate decreasing diagrams for all their local peaks. To this
end, we need to define a labelling schema for sub-rewriting.  Assuming
that rules have an integer index, different rules having possibly the
same index, a step $u \dlrps{p}{\Rsub} v$ with the rule $l_i\ra r_i$
is labelled by the pair $\langle \rank{u|_p}, i\rangle$. Pairs are
compared in the order $\geordl=(\geq_N, \geq_N)_{lex}$ whose strict
part is well-founded.  Indexes give more flexibility (shared indexes
give even more) in finding decreasing diagrams for critical pairs,
this is their sole use.

\begin{definition}
\label{d:cpcheck}
Let $l\lrps{}{i}r, g\lrps{}{j}d\in R$ and $p\in\FPos{l}$ such that
$l|_p=g$ has a solved form $S$. Then, the cyclic critical pair
$\langle r\eta_S,l[d]_p\eta_S \rangle$ has a \emph{cyclic-joinable
  decreasing diagram} if $r\eta_S \dlrder{\langle 1, I\rangle}{R} s
\cc{R_S\eta_S} t \drlder{\langle 1, J\rangle}{R} l[d]_p\eta_S$, whose
sequences of indexes $I$ and $J$ satisfy the decreasing diagram
condition, with the additional condition, in case
$\Var{l[\cdot]_p}\neq\emptyset$, that all steps have a rule index
$k<i$.
\end{definition}

By Corollary~\ref{c:cp}, the ranks of $l\eta_S$ and $l[g]_p{\eta_S}$
are 1. Thanks to rank non-increasingness and Definition~\ref{d:rank},
the cyclic-joinable decreasing diagram --but the congruence closure
part-- is made of terms of rank 1 except possibly $s$ and $t$ which
may have rank $0$. It follows that all redexes rewritten in the
diagram have rank 1. The decreasing diagram condition is therefore
ensured by the rule indexes, which justifies our formulation.

\hide{
By Corollary~\ref{c:cp}, the rank of $l\eta_S$ and $l[g]_p{\eta_S}$ is
1.  Assuming without loss of generality that $\cc{R_S\eta_S}$ is a
joinability diagram, it follows that the cyclic-joinable decreasing
diagram --but the congruence closure part-- is made of terms of rank 1
except possibly $s$ and $t$ which may have rank $0$. It follows that
all terms in the diagram have rank 1. The decreasing diagram
condition is therefore ensured by the rule indexes, which justifies
our formulation.
}

Note further that the condition $\Var{l[\cdot]_p}=\emptyset$ is
automatically satisfied when $p=\rootp$, hence no additional condition
is needed in case of a root overlap. In case where
$\Var{l[\cdot]_p}\neq\emptyset$, implying a non-root overlap, the
additional condition aims at ensuring that the decreasing diagram is
stable under substitution. It implies in particular that there exists
no $i$-facing step. This may look restrictive, and indeed, we are able
to prove a slightly better condition: (i) there exists no $i$-facing
step, and (ii) each step $u\ra^q_k v$ using rule $k$ at position
$q$ satisfies $k<i$ or $\Var{u|_q}\subseteq\Var{g\eta_S}$. We will restrict
ourselves here to the simpler condition which yields a less involved
confluence proof.

We can now state and prove our main result:

\begin{theorem}
\label{t:OLS}
Rank non-increasing layered systems are confluent provided their
realizable cyclic critical pairs have cyclic-joinable decreasing diagrams.
\end{theorem}

\begin{myproof}
Since $\lrps{}{R}\subseteq \lrps{}{\Rsub}\subseteq \dlrder{}{R}$,
$R$-convertibility and $\Rsub$-convertibility coincide. We can
therefore apply van Oostrom's theorem to $\Rsub$-conversions, and
reason by induction on the rank. We proceed by inspection of the
sub-rewriting local peaks $v \rlps{p}{(l\ra r)_R} u \lrps{q}{(g\ra
  d)_R} w$, with $\Var{l}\cap\Var{g}=\emptyset$. We also assume
for convenience that $\Var{l,g}\cap\Var{u,v,w}=\emptyset$. This allows us
to consider $u,v,w$ as ground terms by adding their variables as new constants.
We assume further that variables $x,y\in\Var{l,g}$ become linearized
variables $x^i,y^j$ in $\lin{l}, \lin{g}$, and that $\xi$ is the
substitution such that $\xi(x^i)\!=\!x$ and $\xi(y^j)\!=\!y$, hence
implying $\Var{\lin{l}}\!\cap\!\Var{\lin{g}} \!=\!\emptyset$.

\smallskip
By definition of sub-rewriting, $u|_p=\lin{l}\theta
\lrdertwo{(>\FPos{l})}{R} v'|_p=l\sigma$ and $v=u[r\sigma]_p$, where for
all positions $o\in\Pos{l}$ such that $l|_o=x$ and $\lin{l}|_o=x^i$,
then $x^i\theta \lrdertwo{}{R} x\sigma$.  Similarly, $u|_q=\lin{g}\theta
\lrdertwo{(>\FPos{g})}{R} w'|_q=g\sigma$ and $w=u[d\sigma]_q$, where for
all positions $o\in\Pos{g}$ such that $g|_o=y$ and $\lin{g}|_o=y^j$,
then $y^j\theta \lrdertwo{}{R} y\sigma$.
There are three cases:

\begin{enumerate}
\item
$p\# q$. The case of disjoint redexes is as usual.

\item
$q>p\cdot\FPos{l}$, the so-called ancestor peak case, for which
  sub-rewriting shows its strength. W.l.o.g. we assume $u|_p$ has some
  rank $n+1$ and note that $u|_q$ has some rank $m\le n$ by
  Corollary~\ref{c:rank}.  Since the sub-rewriting steps from $u$ to
  $w$ occur strictly below $p\cdot\FPos{l}$, then $q=p\cdot o\cdot
  q'$ where $l|_o=\xi(y^j)$ and $\lin{l}|_o=y^j$. It follows that
  $w=\lin{l}\tau$ for some $\tau$ which is equal to $\theta$ for all
  variables in $\lin{l}$ except $y^j$ for which
  $\tau(y^j)=\theta(y^j)[d\sigma]_{q'}$.

We proceed as follows: we equalize all $n$-convertible terms
$\{x\sigma:x\!\in\!\Var{r}\}$ in $v$ and
$\{y\tau:y\!\in\!\Var{\lin{l}}\}$ in $w$ by induction hypothesis,
yielding $s, t$. Note that steps from $v$ to $s$ have ranks strictly
less than the rank $n+1$ of the step $u\lrps{}{\Rsub}v$ by
Corollary~\ref{c:rank} and rank non-increasingness. Then, $t$ is an
instance of $l$ by some $\gamma$, and $s$ is the corresponding
instance of $r$, hence $t$ rewrites to $s$ with $l\ra r$. The
equalization steps from $w$ to $t$ have ranks which are not guaranteed
to be strictly less than $m$, hence cannot be kept to build a
decreasing diagram. But they can be absorbed in a sub-rewriting step
from $w$ to $s$ whose first label is at most $n+1$, hence faces the
step from $u\lrps{}{\Rsub}v$: sub-rewriting allows us to rewrite
directly from $w$ to $s$, short-cutting the rewrites from $w$ to $t$
that would otherwise yield a non-decreasing diagram. The proof is
depicted at Figure~\ref{fig:cpeak} (left), assuming $p=\rootp$ for
simplicity. Black color is used for the given sub-rewriting local
peak, blue for arrows whose redexes have ranks at most $n+1$, and red
when redex has rank at most $n$.

\begin{figure}[h]
\setlength{\unitlength}{0.34mm}
\centering
\begin{minipage}[b]{0.45\linewidth}
\begin{picture}(150,125)(-7,5)
\put(-85,43){
\begin{picture}(80,80)
\thicklines

\put(37,78){$u=\lin{l}\theta$}
\textcolor{red}{
\put(36,76){\vector(-3,-2){58}}
\put(36,76){\vector(-3,-2){55}}
}
\put(-32,31){$v'=l\sigma$}
\textcolor{blue}{
\put(-29,28){\vector(0,-1){40}}
}

\qbezier(35,79)(-60,70)(-32,-12)
\put(-32,-12){\vector(1,-3){0}}
\put(-40,65){\footnotesize $n+1$}
\put(-32,-18){$v=r\sigma$}

\put(44,76){\vector(3,-2){60}}
\put(73,58){\footnotesize $m\le n$}
\put(105,31){$w=\lin{l}\tau$}

\put(5,25){Ancestor peak}

%confluence
\textcolor{red}{
\put(-29,-21){\vector(0,-1){40}}
\put(-29,-21){\vector(0,-1){37}}
\put(-27,-38){\footnotesize $\le n$}
}
\put(-32,-68){$s=r\gamma$}

\textcolor{red}{
\put(109,28){\vector(0,-1){88}}
\put(109,28){\vector(0,-1){85}}
}
\put(107,-68){$t=l\gamma$}

\textcolor{blue}{
\put(104,-67){\vector(-1,0){105}}
\qbezier(107,28)(90,-60)(-1,-65)
\put(-1,-65){\vector(-4,-1){0}}
\put(43,-35){\footnotesize $\le n+1$}
}
\end{picture}
}

\put(125,43){
\begin{picture}(80,80)
\thicklines

\put(37,78){$u=\lin{l}\theta=\lin{l}\theta[\lin{g}\theta]_q$}
\textcolor{red}{
\put(36,76){\vector(-3,-2){60}}
\put(36,76){\vector(-3,-2){57}}
}
\put(-32,31){$v'=l\sigma$}
\textcolor{blue}{
\put(-29,28){\vector(0,-1){40}}
}
\qbezier(35,79)(-90,60)(-32,-12)
\put(-32,-12){\vector(2,-3){0}}
\put(-37,68){\footnotesize $n+1$}
\put(-32,-18){$v=r\sigma$}

\textcolor{red}{
\put(44,76){\vector(3,-2){60}}
\put(44,76){\vector(3,-2){57}}
}
\put(105,31){$w'\!=\!\lin{l}\theta[g\sigma]_q$}

\textcolor{blue}{
\put(109,28){\vector(0,-1){40}}
}
\qbezier(44.5,77.5)(230,50)(112,-11)
\put(112,-11){\vector(-3,-2){0}}
\put(110,65){\footnotesize $m\leq n+1$}
\put(105,-18){$w=\lin{l}\theta[d\sigma]_q$}

\put(10,25){Critical peak}

\put(-5,-10){
\begin{picture}(80,80)

\put(14,1){\footnotesize $l\eta_S \cc{R_S\eta_S} l[g]_q\eta_S$}
\put(20,-3){\vector(-1,-1){18}}
\put(60,-3){\vector(1,-1){18}}
\put(-3,-27){\footnotesize $r\eta_S$}
\put(65,-27){\footnotesize $l[d]_q\eta_S$}

\textcolor{green}{
\put(0,-30){\vector(2,-1){32}}
\put(0,-30){\vector(2,-1){30}}
\put(80,-30){\vector(-2,-1){32}}
\put(80,-30){\vector(-2,-1){30}}
\put(33,-50){$\cc{R_S\eta_S}$}
}

\end{picture}}

\textcolor{red}{
\put(-29,-20){\vector(0,-1){38}}
\put(-29,-20){\vector(0,-1){36}}
}
\put(-32,-65){$s=r\gamma$}
\textcolor{red}{
\put(109,-20){\vector(0,-1){13}}
\put(109,-20){\vector(0,-1){11}}
}
\put(107,-41.5){$t'=l[d]_q\sigma$}
\textcolor{red}{
\put(109,-44){\vector(0,-1){13}}
\put(109,-44){\vector(0,-1){11}}
}
\put(107,-65){$t=l[d]_q\gamma$}

\textcolor{green}{
\put(-26,-68){\vector(2,-1){50}}
\put(-26,-68){\vector(2,-1){48}}
\put(106,-68){\vector(-2,-1){48}}
\put(106,-68){\vector(-2,-1){50}}
\put(24.5,-95){$\cc{R_S\eta_S\rho}$}
}

\textcolor{red}{
\put(22,-99){\vector(1,-1){13}}
\put(22,-99){\vector(1,-1){11}}
\put(52.5,-99){\vector(-1,-1){13}}
\put(52.5,-99){\vector(-1,-1){11}}
}
\end{picture}
}
\end{picture}
\bigskip\bigskip\bigskip\bigskip\bigskip\bigskip
\caption{Ancestor and Critical Peaks}
\label{fig:cpeak}
\end{minipage}
\vspace{-2mm}
\end{figure}
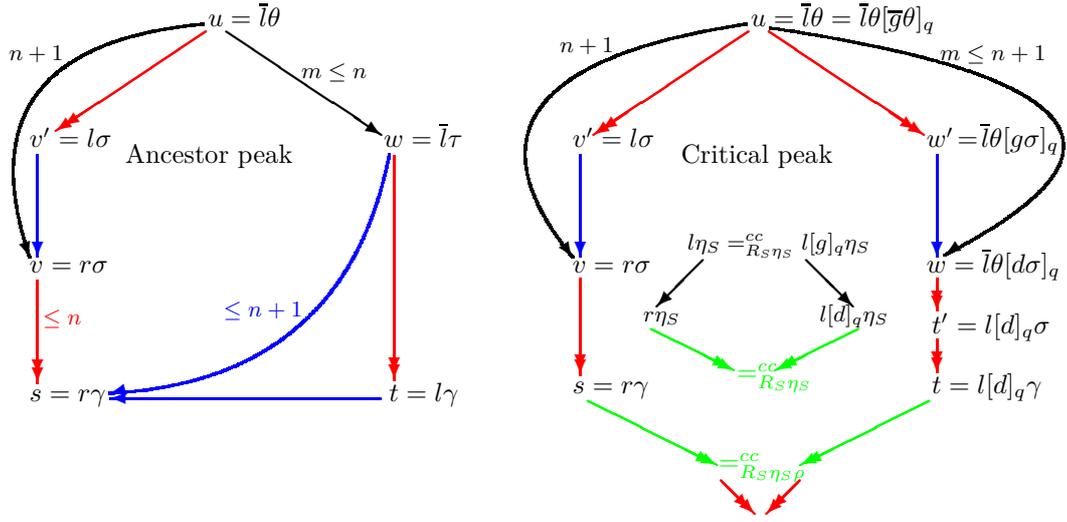

\item
$q\in p\cdot\FPos{l}$, the so-called critical peak case, whose left
  and right rewrite steps have labels $\langle n+1, i\rangle$ and
  $\langle m, j\rangle$ respectively, with rules $l\ra r$ and $g\ra d$
  having indexes $i$ and $j$. Assuming without loss of generality that
  $p=\rootp$, the proof is depicted at Figure~\ref{fig:cpeak}
  (right). Most technical difficulties here originate from the fact
  that the context $l[\cdot]_q$ may have variables. In this case, we
  first rewrite $w$ to $t'=l\sigma[d\sigma]_q=l[d]_q\sigma$ by
  replaying those equalization steps, of rank at most $n$, used in the
  derivation from $u$ to $v'$, which apply to variable positions in
  $\Var{\lin{l}[\cdot]_q}$.

Now, since $\lin{l}\theta=\lin{l}\theta[\lin{g}\theta]_q$, by
Lemma~\ref{l:equalization}, there is a substitution $\gamma$ and a
solved form $S$ of the unification problem $l|_q=g$, such that
$\sigma\lrdertwo{}{R}\gamma$, $\gamma=\eta_S\rho$ for some $\rho$, and
$R_S$ is realizable by $\gamma$. By assumption, the cyclic critical
pair $\langle r\eta_S, l[d]_q\eta_S \rangle$ has a cyclic-joinable
decreasing diagram (modulo $\cc{R_S\eta_S}$).  We can now lift this
diagram to the pair $\langle s,t\rangle$ by instantiation with the
substitution $\rho$. The congruence closure used in the lifted diagram
becomes therefore $\cc{R_S\eta_S\rho}$. We are left showing that the
obtained diagram for the pair $\langle v,w \rangle$ is decreasing with
respect to the local peak $v\la u\ra w$.

This diagram is made of three distinct parts: the equalization steps,
the rewrite steps instantiating the cyclic-joinability assumption with
$\rho$, which originate from $s$ and $t$ --we call them the middle
part--, and the congruence closure steps. By Corollary~\ref{c:rank},
the left equalization steps $v=r\sigma\lrdertwo{}{R}r\gamma=s$ use
rewrites with redexes of rank at most $n$, hence their labels are
strictly smaller than $\langle n+1,i\rangle$. The right equalization
steps $w\lrdertwo{}{}t'\lrdertwo{}{}t$ are considered together with
the (green-)middle-part rewrite steps. There are two cases 
depending on whether $l[\cdot]_q$ is variable-free or not:
\begin{enumerate}
\item
$\Var{l[\cdot]_q}=\emptyset$, hence $m=n+1$ by
  Corollary~\ref{c:rank}. In this case, $w=t'$, and by
  Corollary~\ref{c:rank}, the rewrite steps
  $w=l[d]_q\sigma\lrdertwo{}{R}l[d]_q\gamma=t$ have redexes of rank at
  most $n$, making their labels strictly smaller than $\langle
  m,j\rangle=\langle n+1,j\rangle$. Let us now consider the
  middle-part rewrite steps.  Thanks to rank non-increasingness, all
  terms in this part have rank at most $n+1$. It follows that
  the associated labels are pairs of the form $\{\langle n',i'\rangle
  : n'\le n+1, i'\in I\}$ on the left, or $\{\langle n',j'\rangle :
  n'\le n+1, j'\in J\}$ on the right. The assumption that $I,J$
  satisfy the decreasing diagram condition for the critical peak
  ensures that these rewrites do satisfy the decreasing diagram
  condition with respect to the local peak $v\la u\ra w$ as well.

\item
$\Var{l[\cdot]_q}\neq\emptyset$. By Corollary~\ref{c:rank}, the right
  equalization steps $w\lrdertwo{}{}t'\lrdertwo{}{}t$ have redexes of
  rank at most $n$, making their labels strictly smaller than $\langle
  n+1,i\rangle$. Consider now the middle part.  Thanks to rank
  non-increasingness and the additional condition on the
  cyclic-joinability assumption of the cyclic critical pair, all
  labels $\langle n',k\rangle$ in the middle part satisfy $n'\le n+1$
  and $k<i$, hence are strictly smaller than $\langle n+1,i\rangle$.
\end{enumerate}

\hide{
{\bf advanced version:}\\
with rank non-increasingness, all terms in that part have rank at most
$n+1$. With the additional conditions, each step $\mu\ra^o_k\nu$ using
rule $k$ at position $o$ in the cyclic-joinable decreasing diagram
should satisfy (i)~$\Var{\mu|_o}\subseteq\Var{g\eta_S}$ or (ii)~$k<i$.

- For the steps in the cyclic-joinable decreasing diagram which
satisfy (ii), their corresponding steps in the instantiated diagram
will have labels strictly smaller than $\langle n+1,i\rangle$.

- For each of those which do not satisfy (ii) but satisfy (i), it
should have a rank $k=j$ if it is the $j$-facing step, or a rank $k<j$
otherwise. Then the redex of the instantiated step
$\mu\rho\ra^o_k\nu\rho$ should have rank $\rank{\mu\rho|_o} =
\rank{\mu|_o\rho} = \rank{\mu|_o\rho_{|\Var{g\eta_S}}}$. By
Corollary~\ref{c:rank}, $\rank{\mu|_o\rho_{|\Var{g\eta_S}}} = 1 +
\rank{\rho_{|\Var{g\eta_S}}} = \rank{g\eta_S\rho} =
\rank{g\gamma}$. With rank non-increasingness, $\rank{g\gamma} \le
\rank{g\sigma} \le \rank{\lin{g}\theta} = m$. Hence the label $\langle
n',k\rangle$ of step $\mu\rho\ra^o_k\nu\rho$ satisfies $\langle
n',k\rangle \leordl \langle m,j\rangle$ if $\mu\ra\nu$ is $j$-facing
step, or $\langle n',k\rangle \ltordl \langle m,j\rangle$ otherwise.
}

We are left with the congruence closure steps. Given $y=v\in
R_S$, $y\gamma\lrdertwo{}{R}\rldertwo{}{R}v\gamma$ since $R_S$
is realizable by $\gamma$. By Lemma~\ref{l:unifinvar}, \OF{v} holds,
hence $y\gamma$ and $v\gamma$ are $n$-convertible by rank
non-increasingness. We are left with replacing the
$\cc{\vect{y}\gamma=\vect{v}\gamma}$-steps by a joinability diagram
whose all steps have rank at most $n$.  The obtained diagram is
therefore decreasing, which ends the proof.
\qed
\end{enumerate}
\end{myproof}

Using the improved condition of cyclic-joinability mentioned after
Definition~\ref{d:cpcheck} requires modifying the discussion
concerning the (green-)middle-part rewrite steps. Although this does
not cause any conceptual difficulties, it is technically delicate. The
interested reader can of course reconstruct this proof for
himself or herself.

Our result gives an answer to NKH: confluence of critical pair free 
rewrite systems can be analyzed via their sub-rewriting critical pairs, 
which are actually the cyclic critical pairs.

NKH is critical pair free but non-confluent. Indeed, it has the
$\Omega$ solved form $x=c(y)\land y=c(y)$ obtained by unifying
$f(x,x)=f(y,c(y))$. The cyclic critical peak is then $a \rlps{}{}
f(x,x) \cc{} f(y,c(y)) \lrps{}{} b$ yielding the cyclic critical pair
$\langle a, b\rangle$ which is not joinable modulo
$\{x=c(y),y=c(y)\}$.

We now give a slight modification of NKH making it confluent:
\begin{example}
\label{e:vhuet}
The system 
$R=\{f(x,x)\lrps{}{2} a(x,x),\, f(x,c(x))\lrps{}{2} b(x),\, 
f(c(x),c(x))\lrps{}{3} \\ f(x,c(x)),\, a(x,x)\lrps{}{1} e(x),\, 
b(x)\lrps{}{1} e(c(x)),\, g\lrps{}{0} c(g)\}$ is confluent.
Showing that $R$ satisfies \StOF is routine, and it is rank
non-increasing by Lemma~\ref{l:rkcond}. There are three cyclic
critical pairs, which all have a cyclic-joinable decreasing
diagram. For instance, the unification $f(x,x)=f(y,c(y))$ returns a
canonical cyclic unifier $\langle \eta_S=\emptyset, R_S=\{x\ra c(y),
y\ra c(y)\} \rangle$, the corresponding cyclic critical peak
$a(x,x)\rlps{\langle 1,2\rangle}{} f(x,x) \cc{R_S\eta_S}
f(y,c(y))\lrps{\langle 1,2\rangle}{} b(y)$ has a cyclic-joinable
decreasing diagram $a(x,x)\lrps{\langle 1,1\rangle}{}
e(x)\cc{R_S\eta_S} e(c(y)) \rlps{\langle 1,1\rangle}{} b(y)$.  The
unification $f(x,x)=f(c(y),c(y))$ returns $\langle \eta_S=\{x=c(y)\},
R_S=\emptyset \rangle$, the corresponding (normal) critical peak
$a(c(y),c(y))\rlps{\langle 1,2\rangle}{} f(c(y),c(y))\lrps{\langle
  1,3\rangle}{} f(y,c(y))$ decreases by $a(c(y),c(y))\lrps{\langle
  1,1\rangle}{} e(c(y))\rlps{\langle 1,1\rangle}{} b(y)\rlps{\langle
  1,2\rangle}{} f(y,c(y))$.  By Theorem~\ref{t:OLS}, $R$ is confluent.
\end{example}

Theorem~\ref{t:OLS} can be easily used positively: if all cyclic
critical pairs have cyclic-joinable decreasing diagrams, then
confluence is met. This was the case in Example~\ref{e:vhuet}. But
there is another positive use that we illustrate now: showing that
$\{f(x,x)\ra a, f(x,c(x))\ra b, g\ra d(g)\}$ is confluent requires
proving that the cyclic critical pair given by unifying the first two
rules is not realizable. Although realizability is undecidable in
general, this is the case here since there is no term $s$ convertible
to $c(s)$.  Theorem~\ref{t:OLS} can also be used negatively by
exhibiting some realizable cyclic critical pair which is not joinable:
this is the case of example NKH. In general, if some realizable cyclic
critical pair leading to a local peak is not joinable, then the system
is non-confluent. Whether a realizable cyclic critical pair always
yields a local peak is still an open problem which we had no time to
investigate yet.

A main assumption of our result is that rules may not increase the
rank. One can of course challenge this assumption, which could be due
to the proof method itself. The following counter-example shows that
it is not the case.
%, and therefore, that sub-rewriting is a fundamental concept.

\begin{example}
Consider the critical pair free system $R=\{d(x,x) \ra 0,\,
f(x)\ra d(x,f(x)),\, \\c \ra f(c)\}$, which is layered but whose second
rule is rank increasing since $d(x^1, x^2)$ unifies with $d(y, f(y))$.
This system is non-confluent, since $f(fc) \lrps{}{} d(fc, ffc)
\lrps{}{} d(ffc, ffc) \lrps{}{} 0$ while $f(fc) \lrps{}{}f(d(c, fc))
\lrps{}{} f(d(fc, fc)) \lrps{}{} f0$ which generates the regular tree
language $\{S \ra d(0, S), S\ra f0\}$ not containing $0$. Note that
replacing the second rule by the right linear rule $f(x) \ra d(x,
f(c))$ yields a confluent system~\cite{DBLP:journals/ipl/SakaiO10}.
\end{example}

Releasing rank non-increasingness would indeed require strengthening
another assumption, possibly imposing left- or right-linearity. 
%Rank non-increasingness is considered in Appendix~\ref{s:rni}.

\section{Conclusion}
\label{s:conclusion}
Decreasing diagrams opened the way for generalizing Knuth and Bendix's
critical-pair test for confluence to non-terminating systems,
re-igniting these questions. Our results answer open
problems by allowing non-terminating rules which can
also be non-linear on the left as well as on the right. The notion of
layered systems is our first conceptual contribution here.

Another, technical contribution of our work is the notion of
sub-rewriting, which can indeed be compared to parallel
rewriting. Both relations contain plain rewriting, and are included in
its transitive closure. Both can therefore be used for studying
confluence of plain rewriting. Tait and Martin-L\"of's parallel
rewriting --as presented by Barendregt in his famous book on Lambda
Calculus~\cite{barendregt1984lambda}-- has been recognized as the
major tool for studying confluence of left-linear non-terminating
rewrite relations when they are not right-linear.  We believe that
sub-rewriting will be equally successful for studying confluence of
non-terminating rewrite relations that are not left-linear. In the
present work where no linearity assumption is made, assumption \StOF
ensuring the absence of stacked critical pairs in lefthand sides makes
the combined use of sub-rewriting and parallel rewriting
superfluous. Without that assumption, as is the case
in~\cite{DBLP:conf/rta/LiuDJ14}, their combined use becomes necessary.

A last contribution, both technical and conceptual, is the notion of
cyclic unifiers. Although their study is still preliminary, we have
shown that they constitute a powerful new tool to handle unification
problems with cyclic equations in the same way we deal with
unification problems without cyclic equations, thanks to the existence
of most general cyclic unifiers which generalize the usual notion of
mgu. This indeed opens the way to a \emph{uniform treatment} of
problems where unification, whether finite or infinite, plays a
central role.

Our long-term goal goes beyond improving the current toolkit for
carrying out confluence proofs for non-terminating rewrite systems. We
aim at designing new tools for showing confluence of complex type
theories (with dependent types, universes and dependent elimination
rules) directly on raw terms, which would ease the construction of
strongly normalizing models for typed terms. Since redex-depth, the
notion of rank used here, does not behave well for higher-order rules,
appropriate new notions of rank are required in that
setting. 

\smallskip
{\bf Acknowledgment:} This research is supported in part by NSFC
Programs (No. 91218302, 61272002) and MIIT IT funds (Research and
application of TCN key technologies) of China, and in part by JSPS,
KAKENHI Grant-in-Aid for Exploratory Research 25540003 of Japan.

\end{document}